\documentclass[a4paper,USenglish,cleveref, autoref, thm-restate]{lipics-v2021}

\hideLIPIcs  

\RequirePackage{xspace,xcolor}
\usepackage{xifthen}
\RequirePackage{amsthm}
\RequirePackage{amsmath}
\RequirePackage{amssymb}
\usepackage{algorithm}
\usepackage{mathtools}
\usepackage[noend]{algpseudocode}
\usepackage{algorithmicx}
\algrenewcommand\algorithmicrequire{\textbf{Input:}}
\algrenewcommand\algorithmicensure{\textbf{Output:}} 

\algnewcommand{\IfThenElse}[3]{
  \State \algorithmicif\ #1\ \algorithmicthen\ #2\ \algorithmicelse\ #3}

\usepackage{booktabs}
\usepackage{bussproofs}
\newcommand{\added}[1]{\textcolor{black}{#1}}
\newcommand{\Math}[1]{\ensuremath{#1}\xspace}
\newcommand{\ProofSystemsFont}[1]{\mathsf{#1}}
\newcommand{\ComplexityClassFont}[1]{\mathsf{#1}}
\newcommand{\DefineProofSystem}[2]{\expandafter\def\csname#1\endcsname{\Math{\ComplexityClassFont{#2}}}}

\newcommand{\rqis}{\texttt{EQuIPS}\xspace}
\newcommand{\findMove}{\rqis\xspace}
\newcommand{\findMoveNoX}{\texttt{EQuIPS}}
\newcommand{\yasol}{\texttt{Yasol}\xspace}

\newcommand{\cadical}{\texttt{CaDiCaL}\xspace}
\newcommand{\gurobi}{\texttt{GUROBI}\xspace}
\newcommand{\cplex}{\texttt{CPLEX}\xspace}
\newcommand{\MCNBaggio}{\texttt{MCN}$_{{CR}}$\xspace}
\newcommand{\QIP}{QIP}
\newcommand{\QIPPU}{QIP$_{\textnormal{PU}}$}
\newcommand{\depqbf}{\texttt{DepQBF}\xspace}
\newcommand{\zthree}{\texttt{z3}\xspace}
\newcommand{\extract}{\texttt{extract}\xspace}
\newcommand{\wins}{\texttt{wins1}\xspace} 

\newcommand{\qdimacs}{\texttt{QDIMACS}\xspace}
\newcommand{\qcir}{\texttt{QCIR}\xspace}
\newcommand{\refine}{\texttt{refine}\xspace}
\newcommand{\rqs}{\texttt{RAReQS}\xspace}
\newcommand{\qfun}{\texttt{QFUN}\xspace}
\newcommand{\qfuncms}{\texttt{QFUN}$_{\textnormal{cms}}$\xspace}
\newcommand{\qfunmini}{\texttt{QFUN}$_{\textnormal{cadi}}$\xspace}
\newcommand{\cms}{\texttt{CryptoMiniSat}\xspace}

\newcommand{\D}{\ensuremath \mathcal{D}}
\newcommand{\Q}{\ensuremath \mathcal{Q}}
\newcommand{\U}{\ensuremath \mathcal{U}}
\newcommand{\yvec}{\ensuremath \pmb{y}}
\newcommand{\cvec}{\ensuremath \pmb{c}}
\newcommand{\xvec}{\ensuremath \pmb{x}}
\newcommand{\Xvec}{\ensuremath \pmb{X}}
\newcommand{\Yvec}{\ensuremath \pmb{Y}}
\newcommand{\Zvec}{\ensuremath \pmb{Z}}

\newcommand{\tauvec}{\ensuremath \pmb{\tau}}
\newcommand{\muvec}{\ensuremath \pmb{\mu}}
\newcommand{\betavec}{\ensuremath \pmb{\beta}}

\newcommand{\bvec}{\ensuremath \pmb{b}}
\newcommand{\Lvec}{\ensuremath \pmb{L}}
\newcommand{\rvec}{\ensuremath \pmb{r}}
\newcommand{\y}{\ensuremath \pmb{y}}

\newlength{\breite}
\usepackage{tikz}
\tikzstyle{calcn}=[rectangle%
,rounded corners=1mm%
,thick%
,draw=black,%
,top color=white,bottom color=black!20,%
minimum height=.5cm,minimum width=.5cm,inner sep=1pt%
]
\tikzstyle{qcalcn}=[rectangle%
,rounded corners=2mm%
,thick%
,draw=black!90,%
,top color=white,bottom color=purple!30,%
minimum height=.5cm,minimum width=.5cm,inner sep=1pt%
]

\newcommand{\Red}{\Math{\boldsymbol{\forall}\kern .04ex\ProofSystemsFont{red}}}
\newcommand{\red}{\Math{\,\raisebox{.2ex}{$\scriptstyle+$}\,\Red}}

\newcommand{\Frege}[1][]{\Math{\ifthenelse{\isempty{#1}}{\ProofSystemsFont{Frege}}{#1\text{-}\ProofSystemsFont{Frege}}}}
\DefineProofSystem{eFrege}{e\Frege[]}
\newcommand{\eFregeRed}[1][]{\eFrege\!\!\red}

\newcommand{\ecalculus}{$\forall$\textsf{Exp+Res}\xspace}

\title{An Expansion-Based Approach for Quantified Integer Programming}

\author{Michael Hartisch}{University of Passau, Germany \and FAU Erlangen-Nürnberg, Germany}{michael.hartisch@uni-passau.de}{https://orcid.org/0000-0001-6304-4973}{Partially funded by Deutsche Forschungsgemeinschaft (DFG, German Research
Foundation) - 534441421}

\author{Leroy Chew}{TU Wien, Austria}{lchew@ac.tuwien.ac.at}{https://orcid.org/0000-0003-0226-2832}{funded by Austrian Science Fund FWF Project ESP197.}

\authorrunning{M. Hartisch and L. Chew} 

\Copyright{Michael Hartisch and Leroy Chew} 



\category{} 

\relatedversion{} 

\acknowledgements{We thank the three anonymous reviewers for their helpful comments. Furthermore, we acknowledge that this research was partially funded by Deutsche Forschungsgemeinschaft (DFG, German Research
Foundation) - 534441421 and Austrian Science Fund FWF Project ESP197.}

\nolinenumbers 

\begin{document}

\maketitle

\begin{abstract}
Quantified Integer Programming (QIP) bridges multiple domains by extending Quantified Boolean Formulas (QBF) to incorporate general integer variables and linear constraints while also generalizing Integer Programming through variable quantification. As a special case of Quantified Constraint Satisfaction Problems (QCSP), QIP provides a versatile framework for addressing complex decision-making scenarios. Additionally, the inclusion of a linear objective function enables QIP to effectively model multistage robust discrete linear optimization problems, making it a powerful tool for tackling uncertainty in optimization.

While two primary solution paradigms exist for QBF---search-based and expansion-based approaches---only search-based methods have been explored for QIP and QCSP. We introduce an expansion-based approach for QIP using Counterexample-Guided Abstraction Refinement (CEGAR), adapting techniques from QBF. We extend this methodology to tackle multistage robust discrete optimization problems with linear constraints and further embed it in an optimization framework, enhancing its applicability. Our experimental results highlight the advantages of this approach, demonstrating superior performance over existing search-based solvers for QIP in specific instances. Furthermore, the ability to model problems using linear constraints enables notable performance gains over state-of-the-art expansion-based solvers for QBF.
\end{abstract}

\section{Introduction}

\subsection{Motivation}
Solving Quantified Boolean Formulas (QBF) typically relies on two complementary approaches: Quantified Conflict-Driven Clause Learning (QCDCL) and expansion-based solving. QCDCL extends classical SAT solving techniques by incorporating learning and backjumping in a search-based framework\added{---explicitly traversing the assignment tree respecting the quantifier order---}making it particularly effective for QBF instances with deep quantifier alternations. In contrast, expansion-based solving employs a controlled form of Shannon expansion \added{duplicating the formula to eliminate quantifiers, at the cost of a larger formula with more variables}. While expansion methods excel on formulas with few quantifier alternations, they struggle as the number of alternations increases. QCDCL solvers exhibit the opposite behavior: they handle deeply nested quantifications well but can be inefficient for problems where expansion-based approaches are advantageous.

Quantified Integer Programming generalizes QBF by introducing integer variables and linear constraints, placing it within the broader domain of Quantified Constraint Satisfaction Problems (QCSP). Existing QIP solvers follow a QCDCL-inspired approach. 
 However, no expansion-based QIP solver exists, leaving a gap in the landscape of QIP solution methods. In this work, we address this gap by introducing and investigating a new expansion-based QIP solver, enabling efficient resolution of problems where search-based methods struggle.

Expansion-based methods are refutationally complete: given that QIP requires bounded variable domains, a full expansion reduces satisfiability to a finite Integer linear Programming (IP) problem. However, this naive approach is computationally prohibitive, as the expansion tree grows exponentially with the number of variables. In practice, proving unsatisfiability does not require full expansion---an unsatisfiable core often involves only a fraction of the search space. Inspired by techniques from QBF solving, such as those employed in \rqs and \qfun, we leverage Counterexample Guided Abstraction Refinement (CEGAR) to construct a more targeted expansion. Our method iteratively refines the search space by alternating between satisfiability and unsatisfiability checks, incorporating counterexamples (countermoves) to guide the exploration process efficiently.

By bridging the gap between expansion-based solving and QIP, our approach broadens the applicability of QIP solvers and introduces new avenues for tackling complex quantified constraint problems.

\subsection{Related Work}
The study of quantified integer variables dates back at least to the 1990s \cite{gerber1995parametric}, with the term QIP being coined in \cite{subramani2004analyzing}. Early research focused on the complexity of QIP  \cite{subramani2004analyzing,chistikov2017complexity,nguyen2020computational} and treated QIP as a satisfiability problem until optimization aspects were incorporated \cite{lorenz2015solving}.

Building on these foundations, a search-based approach for solving QIP was proposed \cite{ederer2017yasol,hartisch2022general}, followed by several algorithmic enhancements, including expansion-related  relaxation techniques  \cite{hartisch2021adaptive} and pruning heuristics  \cite{hartisch2020novel}.
Extensions enabled restricting universal variable assignments, confining them beyond their default bounds \cite{hartisch2016quantified,hartisch2019mastering}, similar to robust optimization under polyhedral uncertainty \cite{goerigk2024introduction}, budgeted uncertainty \cite{bertsimas2003robust,goerigk2020min}, or decision-dependent uncertainty \cite{poss2014robust,omer2024combinatorial}. Furthermore, the link to multistage robust discrete optimization has been established, demonstrating practical applicability \cite{goerigk2021multistage}.  

In multistage robust discrete optimization, approaches such as scenario generation \cite{goerigk2024introduction} and row-and-column generation (e.g. \cite{baggio2021multilevel}) iteratively refine relaxed formulations, intrinsically incorporating optimization. Unlike these, in order to optimize, we repeatedly solve QIP feasibility problems in a binary search environment to determine the optimal objective value. For general multistage optimization approaches, most available methods rely on approximation techniques, typically using either sampling methods \cite{bertsimas2016multistage,maggioni2025sampling} or decision rules \cite{vayanos2012constraint,postek2016multistage,bertsimas2018binary}, which restrict the solution space to strategies with a predefined (often linear) structure. 

Our approach builds upon the expansion-based methods developed for QBF \cite{DBLP:conf/sat/JanotaKMC12,JM15,janota2018towards},  contrasting with traditional search-based approaches \cite{DBLP:conf/iccad/ZhangM02,lonsing2017depqbf}, though efforts have been made to combine both ideas \cite{bjorner2015conflicts}. Theoretical~\cite{BCCM18} and practical~\cite{LE18alt} results suggest that expansion-based solving is advantageous when there are few quantifier alternations and several limitations of search-based QBF methods have been documented \cite{BCJ19,BB21,BB23,CM24}.

In QCSP, numerous search-based algorithms exist \cite{mamoulis2004algorithms,gent2008solving,verger2008guiding,nightingale2009non,chen2014beyond}, with only a few approaches employing ideas related to expansion---one repair-based \cite{stergiou2005repair} and one relaxation-based \cite{ferguson2006relaxations} approach. Early attempts to also incorporate optimization date back to the early 2000s \cite{chen2004optimization}. These efforts primarily relied on search-based methods \cite{benedetti2008quantified} and were later extended to more general problem formulations \cite{matsui2010quantified}. Similar developments occurred in QBF, where optimization versions have also been explored \cite{ignatiev2016quantified}.

Finally, several extensions of satisfiability modulo theories (SMT) have been developed to address quantified reasoning \cite{barrett2018satisfiability, bjorner2015playing, reynolds2017solving, barbosa2019extending}. These techniques differ from our QIP-specific approach, which leverages the linear structures for greater efficiency. 
A recent approach in \cite{Thuillier_Siegel_Pauleve_2024} seems related, given the very similar title, but differs in scope, focusing on semi-infinite problems---optimization problems with finitely many (existentially quantified) variables and infinitely many (quantified) constraints.

\section{Preliminaries}\label{sec:prelim}
\subsection{Quantified Integer Programming}
A Quantified Integer Program is a natural extension of an integer linear program in which each variable is subject to either existential or universal quantification. Consider \( n \) integer variables \( x_1, x_2, \ldots, x_n \) arranged in order such that if \( i < i' \), variable \( x_i \), is said to lie to the left of \( x_{i'} \). Each variable \( x_i \), $i\in[n]$\footnote{We use $[n]=\{1,\ldots,n\}$ to denote index sets.}, takes values in a bounded domain $
D_i = \{ x_i \in \mathbb{Z} \mid l_i \leq x_i \leq u_i\}$, where $l_i$ and $u_i$ are the lower and upper integer bounds, respectively. In addition, every variable $x_i$ is associated with a quantifier $Q_i \in \{\exists, \forall\}$. The \textit{quantification level} 
 of a variable is defined as the number of alternations of quantifiers to its left plus one. If there are $k\leq n$ such levels, then all variables sharing the same quantification level are grouped together. For each such level $j$, the common quantifier $\Q_j \in \{\exists, \forall\}$ applies to the vector of variables $\Xvec_j$, which collectively range over the domain $\D_j$ (the Cartesian product of their individual domains). We sometimes omit the level index and simply write $\D$, which implies that the variables must remain integral and adhere to their prescribed bounds.
 
  A \textit{QIP feasibility problem} can be written in the compact form $\Q_1 \Xvec_1 \in \D_1 \; \Q_2 \Xvec_2 \in \D_2 \; \cdots \; \Q_k \Xvec_k \in \D_k : A^\exists \xvec \leq \bvec^\exists$, where the \textit{existential constraint system} is given by  matrix $A^\exists \in \mathbb{Q}^{m \times n}$ and right-hand side vector $\bvec^\exists \in \mathbb{Q}^{m}$, with $\xvec=(\Xvec_1,\ldots,\Xvec_k)$. Note, that throughout the paper we use bold letters, to indicate vectors. We sometimes write $\Q \Xvec \in \D: \Phi$, decomposing the problem into the problem of finding an assignment for the first level variables $\Xvec$ and the remaining QIP $\Phi$ that either starts  with quantifier $\bar{Q}$ or only consists of a constraint system. 
For QIP $\Phi$ and variables $\Xvec$ of  quantification level $j\in[k]$, we write $\Phi[\tauvec]$ to denote the modified QIP obtained by removing $\Xvec$ from the quantification sequence and assigning $\Xvec$ to $\tauvec \in \mathcal{D}_j$ in the constraint system of $\Phi$ leading to a change in the right-hand side vector. Then, the new constraint system is described by $A^\exists_{(-\Xvec)}\xvec_{(-\Xvec)} \leq \bvec^\exists_{\Xvec=\tauvec}$, where $A^\exists_{(-\Xvec)}$ contains all columns as $A^\exists$ without the columns corresponding to $\Xvec$, $\xvec_{(-\Xvec)}$ corresponds to the variable vector $\xvec$ without variables $\Xvec$, and $\bvec^\exists_{\Xvec=\tauvec}=\bvec^\exists -A^\exists_{(\Xvec)}\tauvec$, where $A^\exists_{(\Xvec)}$ is the matrix comprising of the columns of $A^\exists$ corresponding to $\Xvec$.

\subsection{QIP Game Semantics}
We can think of a QIP as a game between the \textit{universal} ($\forall$) and \textit{existential} ($\exists$) player, where in move $j$, player $\Q_j$ assigns variables $\Xvec_j$ of level $j$ with values from domain $\D_j$. The existential player wins, if the existential constraint system is satisfied after all variables have been assigned. The universal player wins, if at least one constraint is violated. More formally, a \textit{play} is an assignment of all variables $\Xvec_1,\ldots,\Xvec_k$. A strategy for the assignment of variables $\Xvec_j$ of level $j$, is a function $s_j:\D_1\times\cdots\times \D_{j-1}\rightarrow \D_j$. A strategy $S=(s_i)_{i\in[k],\Q_i=Q}$ for player $Q$ consists of a strategy for all variables associated to her. For a given QIP, $S$ is a \textit{winning strategy} for the existential player, if
$$
\forall \Xvec_i \in \D_i, i\in [k], Q_i=\forall :\ A^\exists \xvec \leq \bvec^\exists \bigwedge_{\substack{i\in[k]:\\ \Q_i=\exists}} \Xvec_i=s_i(\Xvec_1,\ldots,\Xvec_{i-1})
$$
is true, i.e., for every assignment of the universally quantified variables, the strategy results in a satisfied constraint system. Conversely, $S$ is a winning strategy for the universal player, if

$$
\forall \Xvec_i \in \D_i, i\in [k], Q_i=\exists :\ A^\exists \xvec \not\leq \bvec^\exists \bigwedge_{\substack{i\in[k]:\\ \Q_i=\forall}} \Xvec_i=s_i(\Xvec_1,\ldots,\Xvec_{i-1})
$$
is true, i.e., if regardless of the assignment of the existentially quantified variables, the strategy of the universal player leads to a violation of the constraint system.  

For QIP $\Q \Xvec\in\D: \Phi$, a strategy $s$ for assigning $\Xvec$ (which is just an assignment $\tauvec \in \D$), is called a \textit{winning move}, if there exists a winning strategy $S$ containing $s$. 
Furthermore, for $\Q \Xvec_1\in\D_1 \bar{\Q}\Xvec_2\in \D_2:  \Phi$ and assignment $\tauvec \in \D_1$ to $\Xvec_1$, the assignment $\muvec \in \D_2$ is called a \textit{countermove} to $\tauvec$, if $\muvec$ is a winning move for the QIP $\bar{\Q} \Xvec_2\in\D_2:\ \Phi[\tauvec]$.
We also adopt the notion of a multi-game as used in \cite{janota2018towards}. A QIP \textit{multi-game} is given by $\Q \Xvec \in \D: \{ \Phi_1\ldots \Phi_\ell\}$, where each $\Phi_i$, referred to as a \textit{subgame}, is a QIP beginning with $\bar{Q}$ or that solely contains a constraint system. A multi-game is won by player $\Q$, if there exists a move $\tauvec \in \D$ such that $\Q$ has a winning strategy for all subgames $\Phi_i[\tauvec]$.

\subsection{Additional Constraints}
To further enhance the modeling power of a QIP, we explicitly enable the restriction of universally quantified variables. This has been introduced in \cite{hartisch2016quantified} for QIP and has a similar notion as QCSP$^+$ \cite{verger2008guiding} and Quantified Linear Implication \cite{eirinakis2014quantified}. In robust optimization, this corresponds to introducing an uncertainty set \cite{goerigk2024introduction}, while in the context of QBF it is related to including cubes in the initial formula. To this end, let $A^\forall \xvec \leq \bvec^\forall$
be the \textit{universal constraint system} with $A^\forall \in \mathbb{Q}^{\bar{m} \times n}$ and $\bvec^\forall \in \mathbb{Q}^{\bar{m}}$, $\bar{m}\in\mathbb{Z}_{\geq 0}$. To ensure that the universal player's constraints remain independent of the existential variables, we require that $A^\forall_{\ell,i}=0$ for every $i\in[n]$ with $Q_i=\exists$ and for all $\ell\in [\bar{m}]$. We also assume that $\{\xvec \in \D \mid A^\forall \xvec\leq \bvec^\forall\}\neq\emptyset$, i.e., the ``\textit{uncertainty set}'' is nonempty. In presence of a universal constraint system we speak of a \textit{QIP with polyhedral uncertainty}. 
To this end, we redefine the domains of universally quantified variables, by making them depended on assignments $\tilde{\Xvec}_1,\ldots,\tilde{\Xvec}_{j-1}$ of the preceding levels. Specifically, for $j\in[k]$ with $\Q_j =\forall$, we define (with a slight abuse of notation)
\[
\D_j(\tilde{\Xvec}_1,\ldots,\tilde{\Xvec}_{j-1}) = \left\{ \Yvec \in \D_j \,\middle|\, 
\begin{aligned}
& \exists\, \Xvec_{j+1}\in \D_{j+1},\, \ldots,\, \Xvec_k\in \D_k \text{ such that } A^\forall \xvec \leq \bvec^\forall, \\
& \text{with } \xvec = (\tilde{\Xvec}_1,\ldots,\tilde{\Xvec}_{j-1},\Yvec,\Xvec_{j+1},\ldots,\Xvec_k)
\end{aligned}
\right\}, 
\]
while for existentially quantified variables, the domain remains unchanged. Note that, due to the structure of the universal constraint system, only universal variables from preceding levels can influence this domain.
This definition guarantees that when the universal player selects an assignment $\Yvec$ at level $j$, it obeys the lower and upper bounds and there exists an extension in later levels that results in the satisfaction of the universal constraint system. This way, no play can result in a violation of the universal constraint system. If $\bar{m}=0$, i.e., if there are no universal constraints, using the domain described above, boils down to the bounded domain and we obtain a standard QIP. Hence, for the remainder of the paper, whenever we write $\D$ for a universal domain, it may also be subject to a universal constraint system and we omit stating the dependence on variables of previous levels.


We also allow for a linear objective function $\cvec^\top \xvec$, $\cvec\in\mathbb{Z}^n$, which changes the nature of the game for the existential player, who now has to both satisfy the constraint system and minimize the objective value. Suppose $\Q_1=\exists$ and $\Q_k=\forall$. Then, the \textit{QIP optimization problem} can be stated as
$$
\min_{\Xvec_1 \in \D_1} \max_{\Xvec_2 \in \D_2} \cdots \min_{\Xvec_{k-1} \in \D_{k-1}} \max_{\Xvec_k \in \D_k} \; \cvec^\top \xvec \quad : \quad A^\exists \xvec \leq \bvec^\exists . $$
In particular, this optimization problem is feasible with optimal objective value $z^\star \in \mathbb{Z} $, if and only if the following holds: the QIP feasibility problem
$
\Q_1 \Xvec_1 \in \D_1 \; \Q_2 \Xvec_2 \in \D_2 \; \cdots \; \Q_k \Xvec_k \in \D_k : \; A^\exists \xvec \leq \bvec^\exists \wedge \cvec^\top \xvec \leq z^\star
$
can be won by the existential player, while for the QIP
$
\Q_1 \Xvec_1 \in \D_1 \; \Q_2 \Xvec_2 \in \D_2 \; \cdots \; \Q_k \Xvec_k \in \D_k : \; A^\exists \xvec \leq \bvec^\exists \wedge \cvec^\top \xvec \leq z^\star-1
$
no winning strategy for the existential player exists. The notion of a winning strategy remains the same and we only add the term \textit{optimal winning move} for the existential player, which is the winning move that attains the best worst-case objective. 
\begin{example}
Let us consider an example with $n=5$ variables, with domains $D_1=\{0,1,2\}$ and $D_i=\{0,1\}$, for $i\in\{2,3,4,5\}$. The quantification sequence (without the bounding domains for sake of presentation) is given by $\exists x_1 \forall x_2 \exists x_3 \forall x_4 \exists x_5$ and we have a universal constraint system consisting of a single constraint $x_2+x_4\leq 1$. Hence, $\D_2=\{0,1\}$, $\D_4(x_2=0)=\{0,1\}$, and $\D_4(x_2=1)=\{0\}$.  Let the QIP optimization problem be given by
$$\min_{x_1 \in \D_1} \max_{x_2 \in \D_2} \min_{x_3 \in \D_3} \max_{x_4 \in \D_4} \min_{x_5 \in \D_5}  -  x_1 + 2 x_2 - 3 x_3 + x_4 + 2 x_5 :
\begin{matrix}
2x_1 & &+x_3 & &-x_5 &\leq 4\\
\hfill x_1 &-x_2 &+x_3 & &-x_5 &= 1\\
&+x_2 &+x_3 &-x_4 &-x_5 &\leq 2\\
\hfill x_1 & +x_2&+x_3&+x_4&&\leq 3.
\end{matrix}
$$

Here, $x_1=0$ is not a winning move, as there is a countermove $x_2=1$ which renders the second constraint violated. 
There are two winning strategies $S=(s_1, s_3, s_5)$ and $T=(t_1, t_3, t_5)$ for the existential player: $s_1=1$, $s_3=1$, and $s_5=1-x_2$ as well as  $t_1=2$, $t_3=0$, and $t_5=1-x_2$. Note, that if the universal constraint was not present, neither of them would be a winning strategy, as the fourth existential constraint is always violated in case of universal strategy $x_2=x_4=1$.

The worst-case objective value of $T$ is equal to $1$, stemming from the worst case scenario $x_2=0$, $x_4=1$. The same universal assignment also defines the worst-case for $S$, resulting in an objective value of $-1$. Hence, $S$ is a better strategy than $T$ and in fact $x_1=1$ is the optimal winning move. To further clarify the used notation, $A^\exists_{(-x_1)}\xvec_{(-x_1)} \leq \bvec^\exists_{x_1=1}$ is 
$$
\begin{matrix}
 &+x_3 & &-x_5 &\leq 2\\
-x_2 &+x_3 & &-x_5 &= 0\\
+x_2 &+x_3 &-x_4 &+x_5 &\leq 2\\
 +x_2&+x_3&+x_4&&\leq 2,
\end{matrix}
$$
which is the constraint system of the remaining QIP $\Phi[\tau]$ after making the optimal winning move $\tau=x_1=1$.

\end{example}

\section{Expansion-Based Quantified Integer Programming Solver}
\subsection{The Framework}
In this section we present \added{our novel solution approach} the \textit{Expansion-based Quantified Integer Programming Solver} (\rqis) that is able to solve QIP problems with polyhedral uncertainty\footnote{available at \url{https://github.com/MichaelHartisch/EQuIPS}
}.
\rqis is based on an algorithm that iteratively solves abstractions until convergence is achieved. The idea is essentially the same as in the solvers \rqs and \qfun and our pseudo-code is based on the one shown in \cite{janota2018towards}. \added{We adapt their QBF approach to general integer domains and linear functions by accounting for semantic and technical differences, introducing an expansion rule tailored to QIP, and defining a refinement step specific to QIP.} The solution process starts with an empty abstraction of the full QIP---a trivial problem where only compliance with variable domains (possibly including universal constraints) have to be followed---and recursively refines the abstraction by adding found countermoves. This way, it partially expands the inner quantifiers until it is sufficient to solve the original QIP. The pseudocode in Algorithm~\ref{alg:br} shows the main function that calls itself in a nested fashion.  Section~\ref{app::SoundAndComplete} gives a formal proof on correctness and insights regarding the underlying proof system.

\begin{algorithm}
	\begin{algorithmic}[1]
	\Require{multi-game $\Q \Xvec \in \D: \{\Phi_1 \dots \Phi_\ell$\}}
\Ensure{assignment of $\Xvec$ that wins the multi-game or $\bot$ if no winning move exists}
		\If {each $\Phi_l$ is quantifier free}
		\Return $\wins(\Q \Xvec\in \D: \{\Phi_1 \dots \Phi_\ell\})$ \label{line::wins}
		\EndIf
		\State $\alpha \gets \Q \Xvec\in \D: \varnothing$ \Comment{start with empty abstraction} \label{line::empty}
		\While{\textsf{True}}
		\State $\tauvec' \gets \findMove(\alpha)$ \Comment{find winning move for abstraction} \label{line::find_tau}
		\If {$\tauvec'=\bot$}
		\Return $\bot$ \label{line::bot} \Comment{no move for abstraction $\Rightarrow$ no move for multi-game}
		\EndIf
		\State $\tauvec \gets \extract(\tauvec',\Xvec)$\label{line::extract} 
		\State $\lambda \gets -1$
		\For{$l \in [\ell]$ \label{line::start_counter}}
        \State $\muvec \gets$ \findMove($\Phi_l[\tauvec]$) \Comment{find countermove to $\tauvec$}
		\If{ $\muvec \neq \bot$}
		 $\lambda \gets l$
		\EndIf
	\EndFor
		\IfThenElse{$\lambda=-1$}{\Return $\tauvec$}{$\alpha \gets $\refine($\alpha, \Phi_l, \muvec$)\label{line::end_counter}\Comment{refine abstraction}}


		\EndWhile
	\end{algorithmic}
	\caption{\findMoveNoX$(Q\Xvec\in \D: \{\Phi_1 \dots \Phi_\ell\})$---Find Winning Move for Multi-Game\label{alg:br}}
\end{algorithm}

The initial input of Algorithm~\ref{alg:br} is the QIP $\Q \Xvec\in \D: \Phi$, i.e., a multi-game with a single subgame. In Line~\ref{line::wins}, we deal with the case where each subgame of the multi-game is quantifier free. It is noteworthy, that the \wins function call significantly differs from the one used in \cite{janota2018towards}. The main reason is that in case of QBF and $\Q=\forall$,  a winning move for the universal player can be found by solving a SAT problem consisting of the conjunction of the negated Boolean formulas of each subgame. Recall, that the goal for the existential player is to identify a move that ensures each constraint system of every subgame to be  violated. In the case of QIP, it is not immediately clear what the counterpart of a negated Boolean formula would be for a system of linear constraints. We discuss the \wins function in more detail in Section~\ref{sec::wins}. 

In Line~\ref{line::empty} of Algorithm~\ref{alg:br} we initialize the abstraction, containing no subgames. 
A move that wins the current abstraction is found in (Line~\ref{line::find_tau}). 
If $\alpha$ is the empty abstraction, this call to \rqis will immediately invoke \wins, returning any assignment that satisfies the domain $\D$. When $\Q=\exists$, any assignment from the bounded domain may be returned. However, when $\Q=\forall$, compliance with the universal constraint system must also be ensured.

A technical detail to note is that the abstraction may include copies of later-stage variables (due to the subsequent \refine call), necessitating the extraction of only those assignments corresponding to the variables of interest, $\Xvec$, in Line~\ref{line::extract}. After obtaining the corresponding move $\tauvec$, in case a countermove $\muvec$ exists, the abstraction is refined by adding the subgame that corresponds to $\muvec$ (Lines~\ref{line::start_counter}--\ref{line::end_counter}). We will provide more insights into the \refine function in Section~\ref{sec::refine}. If we are not able to find a move that wins the abstraction, we know by construction, that there cannot exist a winning move for the initial multi-game, and return $\bot$ in Line~\ref{line::bot}. Similarly, if no countermove to $\tauvec$ can be found, this means that $\tauvec$ is not only a winning move for the abstraction but also for the entire multi-game. In this case, we return $\tauvec$ (see Line~\ref{line::end_counter}).

\subsection{\wins on Integer Linear Programs\label{sec::wins}}

One crucial aspect of our algorithm is the call $\wins(\Q \Xvec \in \D: \{\Phi_1 \dots \Phi_\ell\})$, where all subgames $\Phi_1 \dots \Phi_\ell$ are quantifier free, i.e., they each only contain a constraint system with variables $\xvec=\Xvec$. Simply speaking, this call tries to answer the question, whether player $\Q$ can find an assignment of $\Xvec$, which is a winning move for all subgames $\Phi_1, \dots, \Phi_\ell$. 
\added{The notation of $\wins$ is borrowed from \qfun, and in \qfun, it is more or less a call to a SAT solver. Our subroutine of \wins, however requires some non trivial modification.}
In the case of a QBF, this problem can be stated in a straight-forward manner as if $\Q=\exists$ one has to find an assignment satisfying the conjunction of the $\ell$ copies of the Boolean function, while if $\Q=\forall$ a solution to the conjunction of all negated Boolean formulas must be found. But in particular regarding the latter case, there is no counterpart in integer programming: there is no notion of a ``negated constraint system''.

For $l\in [\ell]$, let $A_l^\exists \xvec \leq \bvec_l^\exists$ be the constraint system corresponding to subgame $\Phi_l$.
Recall, that if $\Q=\exists$, a winning move  ensures that the constraint system of each subgame is satisfied. Hence, in order to find a winning move, we need to find a solution to Problem~\eqref{prob::ExistEval}:
\begin{subequations}
\label{prob::ExistEval}
\begin{align}
& A_l^\exists\xvec\leq \bvec^\exists_l \quad \forall l\in[\ell]\\
&\xvec \in \D
\end{align}
\end{subequations}
This is a standard integer program and can be solved using any standard solver.

If $\Q=\forall$, we need to find an assignment of $\xvec \in \D$, such that all constraint systems are violated, while obeying the own domain, i.e., the uncertainty set. Violating a constraint system means that at least one of its constraints is not satisfied. To the end of modeling this as an integer linear program, let $\Lvec^l \in \mathbb{Q}^{m}$ be a vector for each $l \in [\ell]$, where  
$$
L^l_j \leq \min_{\xvec\in\D}(A^\exists_l)_{j,\star}\xvec = \sum_{\substack{i\in [n] \\ (A^\exists_l)_{j,i}<0}} (A^\exists_l)_{j,i}u_i + \sum_{\substack{i\in [n] \\ (A^\exists_l)_{j,i}\geq0}} (A^\exists_l)_{j,i}l_i
$$   
for every $ j \in [m] $. In other words, the $j$-th entry of $\Lvec^l$ provides a lower bound that is not larger than the minimum possible value of the left-hand side of row $j$ in constraint system $l$. As a result, the inequality $\Lvec^l \leq A^\exists_l \xvec $ is trivially fulfilled for any $\xvec \in \D$.  

These lower bounds only need to be computed once at the start of our solver since the constraint system remains (basically) the same across all subgames, only differing in the values of already assigned variables. One also could refine these bounds dynamically for each subgame by considering already assigned variables, aiming to accelerate the IP solution process through potentially improved relaxations. However, this approach comes at the cost of additional computational effort, as the bounds must be recomputed at each invocation of \wins. In our implementation, we opted for the weaker bounds that only need to be calculated once.  

Furthermore, we need to be able to certify a violation of a constraint. To this end, we need the following lemma.

\begin{lemma}
Given a linear constraint $\sum_{i\in[n]}a_i x_i \leq b$ with rational coefficients $a_i, b \in\mathbb{Q}$ and integer variables $x_i$. Then, for any integer assignment $\tilde{\xvec}$, it holds that 
$\sum_{i\in[n]}a_i \tilde{x}_i \not\leq b \Leftrightarrow \sum_{i\in[n]}a_i \tilde{x}_i \geq b +r$, where $r=\frac{1}{d}$, for $d=lcd\{a_1,\ldots,a_n,b\}$, being the reciprocal of the lowest common denominators of the $a_i$ and $b$. 
\end{lemma}

\begin{proof}
Let $d$ be the lowest common denominator of the $n+1$ parameters. Then the constraint can be rewritten as 
$\sum_{i\in[n]}\frac{\tilde{a}_i}{d} x_i \leq \frac{\tilde{b}}{d}$, with integers $\tilde{a}_i$ and $\tilde{b}$. For any assignment $\tilde{\xvec}$ with $\sum_{i\in[n]}a_i \tilde{x}_i > b$ the gap between the right-hand side and the left-hand side can be stated as  $\left\vert\frac{\tilde{b}-\sum_{i\in[n]}\tilde{a}_ix_i}{d}\right\vert>0$, where the numerator is integer. Hence, a lower bound for this gap is attained if the numerator is equal to one, i.e., $\frac{1}{d}$ is a lower bound on the violation of the constraint, which proves the claim. 
\end{proof}
\added{This lemma shows the need for integrality of universally quantified variables in our approach, as otherwise, we would not be able to specify a value of minimal violation, as a continuous variable could violate the right-hand side by an arbitrarily small value. For existentially quantified variables this in principle is not necessary in our current setting.}
Now, let $\rvec^l\in\mathbb{Q}^{m}$ be a vector with positive entries less than or equal to the reciprocals of the lowest
common denominators of the rows of $A^\exists_l$ and $\bvec^\exists_l$,
 ensuring for any row $j\in[m]$ and any $\xvec\in \D$ that $(A^\exists_l)_{j,\star}\xvec \not \leq (\bvec^\exists_l)_k \Leftrightarrow  (A^\exists_l)_{j,\star}\xvec \geq (\bvec^\exists_l)_k +r^l_j$. Note that the lowest common denominator of the $n+1$ rational numbers can be computed in polynomial time, using the Euclidean algorithm, assuming that the denominators of the coefficients are know. But also note that smaller values to bound the gap between right-hand side and left-hand side are allowed. E.g.\ for the constraint $0.5x_1+2x_2<=4$ the designated value would be $0.5$, as the lowest common denominator of $0.5$, $2$ and $4$ is $2$ with a reciprocal of $0.5$. On the other hand, finding the number with the most decimal places also is valid. Let $p$ be this number. Then setting $r^l_j=10^{-p}$ also suffices. In the above example this equates to setting $r^l_j=0.1$. The latter case is what we implemented. Further note, that that if all coefficients are integers, this value always can be set to one. 
  
Now, consider the following integer linear program~\eqref{prob::AllEval}:

\begin{subequations}
\centering
\label{prob::AllEval}
\begin{align}
	&-A^\exists_l\xvec-(\Lvec^l-\bvec^\exists_l-\rvec^l)\y^l\leq -\Lvec^l&&\forall l\in[\ell]\label{eq::2a}\\
	&v_l\leq \sum_{j\in[m]}y^l_j&&\forall l\in[\ell]\\
	&v\leq v_l &&\forall l\in[\ell]\\
    & v \geq 1\\
	& A^\forall \xvec \leq \bvec^\forall\\
	& \xvec \in \D\\
	& v, v_1, \ldots, v_\ell \in \{0,1\}\\
	& \yvec^l\in \{0,1\}^{m}&&\forall l\in[\ell]
\end{align}
\end{subequations}
The idea is that the indicator variable $v$ only can be set to 1, if all $\ell$ systems are violated by an assignment of $\xvec$. We can immediately see, that in order to set $v$ to $1$, all $v_l$ must be set to $1$. Hence, for any $l\in[\ell]$ there has to exist at least one constraint $j\in[m]$ for which $y_j^l=1$. So let us consider the setting of $y_j^l$. Setting $y_j^l=0$ can never result in a violation of the respective constraint, as $A^\exists_l \xvec\geq \Lvec^l$ is always true. On the other hand, it is only feasible to set $y_j^l$ to $1$, if $(A^\exists_l)_{k,*}\xvec \geq (b^\exists_l)_j + r^l_j$ holds, which is only the case, if $\xvec$ violates constraint $j$ of system $l$. Consequently, $v_l$ is an indicator whether constraint system $l$ is violated. Only if we find some $\xvec$ that violates all constraint systems while at the same time obeys the uncertainty set given by $A^\forall \xvec \leq \bvec^\forall$, we can also set $v=1$. In other words: If and only if there exists an assignment of $\xvec\in\D$ with $A^\exists_l\xvec\not\leq \bvec_l^\exists$ for all $l\in [\ell]$, Problem~\eqref{prob::AllEval} has a feasible solution. 
\begin{example} \added{Consider $0.5x_1+2x_2<=4$ with bounds $-2 \leq x_1 \leq 2$ and $0\leq x_2 \leq 3$. The lower bound $L$ for the left-hand side is $-1$. We set $r=0.5$ to indicate the minimal violation of this constraint. Consequently, the corresponding Constraint~\eqref{eq::2a} is given by $-0.5x_1-2x_2+5.5y \leq 1$ and in particular, in case of assigning $y=1$, $0.5x_1+2x_2\geq 4.5$ must be true, indicating the violation of the original constraint. }
\end{example}

The pseudocode of the \wins function is presented in Algorithm~\ref{alg:wins}. In our implementation we utilize the solver \gurobi \cite{gurobi} to solve Problems~\eqref{prob::ExistEval} and \eqref{prob::AllEval}.

\begin{algorithm}[htb]
\begin{algorithmic}[1]
	\Require{multi-game $\Q \Xvec \in \D:\{ \Phi_1 \dots \Phi_\ell\}$ with all $\Phi_l$ quantifier free}
\Ensure{assignment of $\Xvec$ that wins the multi-game or $\bot$ if no winning move exists}
\IfThenElse{$\Q = \exists$}{$\pi\gets$ Problem~\eqref{prob::ExistEval}}{$\pi\gets$ Problem~\eqref{prob::AllEval}}
\IfThenElse{$\pi$ is feasible}{\Return solution on $\pi$}{\Return $\bot$}
	\end{algorithmic}
	\caption{$\wins(\Q \Xvec \in \D:\{ \Phi_1 \dots \Phi_\ell\})$---Solve Multi-Game with a Single Move\label{alg:wins}}
\end{algorithm}

\subsection{Refinement by Expansion\label{sec::refine}}

It can be argued that the unique aspect of \rqs and \qfun that sets it apart from other CEGAR approaches is that they extend the number of variables being looked at in the abstraction via expansion.
This is much easier to spot in the original \rqs than in \qfun (our description in  Algorithm~\ref{alg:br} is based off \qfun). In \qfun and our description this is achieved by having multiple listed subproblems after the outer quantifier block, as these subgames technically each have separate variables. The function \refine adds an additional subgame to the abstraction, based on the countermove $\muvec$, that was found to beat our move $\tauvec$ in one of the original subgames.

Given an abstraction $\alpha=\Q \Xvec \in \D_{\Xvec} : \{ \Psi_1, \dots, \Psi_n\}$, with subgames $\Psi_i$, $i\in[n]$, for which a winning move $\tauvec$ was found. Let $\Phi$ be another subgame, in which $\tauvec$ is not a winning move, i.e., there exists a countermove $\muvec$ that wins $\Phi[\tauvec]$. Let the quantification sequence of $\Phi$ start with $\bar{\Q}
\Yvec$. If $\Phi=\bar{\Q}\Yvec \in \D_{\Yvec}:A^\exists \xvec \leq \bvec^\exists$, then

$$\refine(\alpha,\Phi,\muvec)= \Q \Xvec\in D_{\Xvec}: \left\lbrace\Psi_1, \dots , \Psi_n, A^\exists_{(-\Yvec)}\xvec_{(-\Yvec)} \leq \bvec^\exists_{\Yvec=\muvec}\right\rbrace,$$
i.e., the refined abstraction contains an additional constraint system accounting for the scenario of the found countermove.
If  $\Phi=\bar{\Q}\Yvec\in\D_{\Yvec} \Q \Zvec\in \D_{\Zvec}: \Lambda$ for a QIP $\Lambda$, then
$$\refine(\alpha,\Phi,\muvec)= \Q \Xvec \in \D_{\Xvec} \Zvec^{(\Yvec=\muvec)}_{\Phi} \in \D_{\Zvec} :\left\lbrace\Psi_1, \dots , \Psi_n,  \Lambda^{(\Yvec=\muvec)}_{\Phi}[\muvec]\right\rbrace,$$ where $\Zvec^{(\Yvec=\muvec)}_{\Phi}$ is a copy of $\Zvec$ that represents the move of $\Zvec$ in case $\Yvec$ is set to $\muvec$, and the variables of $\Lambda^{(\Yvec=\muvec)}_{\Phi}$ are copies of the variables of $\Lambda$ having the same annotation as $\Zvec^{(\Yvec=\muvec)}_{\Phi}$. When $\Zvec$ is universal, we also need to make sure it satisfies the uncertainty set and thus $\D_{\Zvec}$ is meant to include the respective constraints on the annotated variables. 
We sometimes write $\Zvec^{(\muvec)}$ instead of $\Zvec^{(\Yvec=\muvec)}$.



%
%
%

\begin{example}
	Consider the QIP
	$\exists  x_1 x_2 \forall z_1 z_2 \exists t  d : (x_1+x_2+t-2d=0)\wedge(z_1+z_2+t\geq 1)\wedge(-z_1-z_2-t\geq -2)$
with all binary domains. Consider the outer level. Initially $\alpha$, the abstraction, will be empty, which means that any binary assignment of the variables $x_1$ and $x_2$ is feasible. If we choose $\tauvec=(0,0)$, the universal response is to assign $z_1=z_2=0$, at which point the universal player wins. Hence, a countermove $\muvec=(0,0)$ to $\tauvec$ is found and consequently the abstraction $\alpha$  is refined to become $\alpha= \exists x_1 x_2 t^{(00)} d^{(00)}: (x_1+x_2+t^{(00)}-2d^{(00)}=0) \wedge (t^{(00)}\geq 1)\wedge(-t^{(00)}\geq -2)$. The $\rqis$ call on this refined abstraction will end up in a call of the \wins function, as after the initial existential quantifier, no further quantifiers follow. A solution to the respective constraint system is $\tauvec'= (0, 1, 1, 1)$, which contains assignments of $x_1$, $x_2$, $t^{(00)}$, and $d^{(00)}$. Extracting the values of the relevant variables $x_1$ and $x_2$ we obtain $\tauvec= (0,1)$. We again check whether we find a countermove to $\tauvec$, in which case $z_1=z_2=1$ is produced. 
We once again adapt the abstraction by calling $$\refine(\alpha, \forall z_1 z_2\exists t d: (x_1+x_2+t-2d=0)\wedge(z_1+z_2+t\geq 1)\wedge(-z_1-z_2-t\geq -2), (1,1)),$$ yielding the refined abstraction 
\begin{align*}
\hspace{-0.5cm}\exists x_1 x_2 t^{(00)} d^{(00)} t^{(11)} d^{(11)}: &
\left\lbrace\left( (x_1+x_2+t^{(00)}-2d^{(00)}=0) \wedge (t^{(00)}\geq 1)\wedge(-t^{(00)}\geq -2)\right)\right., \\ & \left.\left( (x_1+x_2+t^{(11)}-2d^{(11)}=0) \wedge (t^{(11)}\geq -1)\wedge(-t^{(11)}\geq 0) \right) \right\rbrace ,
\end{align*}
with two subgames. As each subgames is quantifier free, another call to \wins is invoked and the IP solver is called on the constraint system
\begin{align*}
x_1+x_2+t^{(00)}-2d^{(00)}&=0	&	t^{(00)}&\geq1 	&	-t^{(00)}&\geq-2\\
x_1+x_2+t^{(11)}-2d^{(11)}&=0	&	t^{(11)}&\geq -1	& -t^{(11)}&\geq 0\\
x_1, x_2 ,t^{(00)}, d^{(00)}, t^{(11)}, d^{(11)}&\in\{0,1\}.
\end{align*}

As this IP is infeasible, we know that there is no move for $(x_1,x_2)$ that wins the abstraction, and hence, there cannot exist a move for $(x_1,x_2)$ that wins the game.

\end{example}

\subsection{Underlying Proof System and Correctness\label{app::SoundAndComplete}}

It has been well established that \rqs works with the QBF proof system \ecalculus \cite{JM15}. It comes as no surprise that the underlying proof system of \rqis acts much the same, instead of describing the SAT oracle as a resolution system, we describe the IP call as a cutting planes proof. Figure~\ref{fig:psystem} describes the $\forall$Exp+Cutting Planes proof systems which is the underlying proof system for refutations where the first quantifier is existential.

\begin{figure}[h]
	\framebox{\parbox{0.9\textwidth}{
			\begin{prooftree}
				\AxiomC{{$\displaystyle\sum_{k\in [n]:\ Q_k=  \exists} a_k x_k + \displaystyle\sum_{k\in [n]:\ Q_k=  \forall} a_k x_k \leq b$}\text{ in matrix}}
				\RightLabel{(Axiom)}
				
				\UnaryInfC{{$\displaystyle\sum_{k\in [n]:\ Q_k=  \exists} a_k x^{[\tauvec]}_k + \displaystyle\sum_{k\in [n]:\ Q_k=  \forall} a_k \tauvec(x_k) \leq b $}}
			\end{prooftree}
			\begin{itemize}
				\item[-] {$\tauvec$} is a {complete} assignment to universal variables that satisfies the universal constraint system\\
				\item[-] For $x_k^{[\tauvec]}$, $[\tauvec]$ takes only the part of $\tauvec$ that is left of $x_k$
			\end{itemize}
			\begin{itemize}
				\item {\bf Addition}: From $ \displaystyle\sum_{k\in [n]}
				a_k x_k \leq b$ and $ \displaystyle\sum_{k\in [n]} \alpha_k
				x_k \leq \beta$, derive $ \displaystyle\sum_{k\in [n]} (a_k
				+ \alpha_k) x_k \leq b + \beta$.
				
				\item {\bf Multiplication}: From
				$ \displaystyle\sum_{k\in [n]} a_k x_k \leq b$, derive
				$ \displaystyle\sum_{k\in [n]} da_k x_k \leq db$, where
				$d \in \mathbb{Z}^+$.
				
				\item {\bf Division}: From $ \displaystyle\sum_{k\in [n]}
				a_k x_k \leq b$, derive $\displaystyle\sum_{k\in [n]}
				\frac{a_k}{d} x_k \leq \left\lceil \frac{b}{d} \right\rceil$,
				where  $d \in \mathbb{Z}^+$ divides each $a_k$.
			\end{itemize}
	}}
	\caption{The proof system $\forall$Exp+Cutting Planes \label{fig:psystem}}
\end{figure}

\begin{theorem}\label{thm::ProofSystem}
	$\forall$Exp+Cutting Planes is a sound and complete proof system for QIP.
	
\end{theorem}

\begin{proof}
	Given a QIP, we take a full Shannon expansion on all universal quantifiers. In case a universal constraint system is present, only universal variable assignments satisfying this system are considered.  
	Every potential axiom line is found as a conjunct in this expansion. The full expansion is satisfiability equivalent with the original QIP and contains no universal quantifiers. Therefore given the completeness and soundness of the Cutting Planes proof system, we prove $\forall$Exp+Cutting Planes is a sound and complete proof system for QIP.
\end{proof}

We will now give an overview on the soundness of \rqis. We start with an observation about Algorithm~\ref{alg:br} and the description of refinement in Section~\ref{sec::refine}.

\begin{observation}\label{obs:abs}
	At any point in a run of \rqis with outer block $\Q \Xvec \in D_{\Xvec}$, for every subgame of the abstraction there is an assignment $\muvec$ to the first inner block variables $\Yvec$ such that the subgame is of one of two forms:
	
	\begin{enumerate}
		\item $\Psi_i= A^\exists_{(-\Yvec)}\xvec_{(-\Yvec)} \leq \bvec^\exists_{\Yvec=\muvec}$ when $\Phi_{l_i}= \bar \Q \Yvec\in \D_{\Yvec}: A^\exists \xvec \leq \bvec^\exists$, or 
		\item $\Psi_i= \Lambda^{(\Yvec=\muvec)}_{\Phi_{l_i}}[\muvec]$ when $\Phi_{l_i}= \bar \Q \Yvec\in \D_{\Yvec} \Q \Zvec \in \D_{\Zvec}: \Lambda$, where $\Lambda$ itself is a QIP.
	\end{enumerate}

	In the first case we have a QIP with only an outer block, in the second case we have a QIP with at least one inner block, $\Lambda$ is a QIP representing the rest of blocks and the constraints.
	
	
\end{observation}

\begin{lemma}
	If \rqis returns $\bot$ on multi-game $\forall \Xvec \in \D_{\Xvec}:\{ \Phi_1 \dots \Phi_n\}$ then this multi-game is won by the existential player. 
\end{lemma}

\begin{proof}
	We prove this by induction on the quantifier depth.
	The base case is if $\Phi_1 \dots \Phi_n$ are quantifier free, and we solve integer linear program~\eqref{prob::AllEval} as \wins is invoked. Then, $\bot$ is returned if and only if the universal player is not able to violate the constraint systems of all subgames as argued in Section~\ref{sec::wins}. Therefore $\forall \Xvec \in \D_{\Xvec}:\{ \Phi_1 \dots \Phi_n\}$ is won by the existential player.

	Now suppose some $\Phi_j$ contains a quantifier. Then, we build an abstraction $\alpha$ of the multi-game and $\bot$ is only returned, if the call to \rqis in Line~\ref{line::find_tau} on the abstraction returns $\bot$. Note that the first call to Line~\ref{line::find_tau} on the empty abstraction always returns a move $\tauvec$, as $\D_{\Xvec}$ is always non-empty. Thus, a refinement of the abstraction must have led to a returned $\bot$.
	Hence, we have to show that a call to \rqis for a refined abstraction $\forall \Xvec\in \D_{\Xvec} \forall \Zvec^{(\Yvec=\muvec_1)}_{\Phi_{l_1}}\in \D_{\Zvec}  \dots \Zvec^{(\Yvec=\muvec_k)}_{\Phi_{l_k}}\in \D_{\Zvec} : \{\Psi_1 \dots \Psi_k\}$ returns $\bot$, only if the original formula is won by the existential player.
	We know that the abstraction is existentially feasible by induction hypothesis, in other words there is no assignment to the outer block that makes all subgames infeasible.
	We now argue that for any $\tauvec\in\D_{\Xvec}$, there must be one subgame that is won by the existential player. Otherwise 
	for each $\Phi_j$, $j\in[n]$ there exists an assignment $\tauvec_i\in \D_{\Zvec}$ for which $\Psi_i[\tauvec][\tauvec_i]$ is lost for the existential player. Then we can construct $\tauvec' = (\tauvec,\tauvec_{l_1},\ldots,\tauvec_{l_k})$. This is well defined because $\Xvec$ (corresponding to $\tauvec$) are the only shared outer variables between the subgames. $\tauvec'$ is also a winning move of the abstraction and in particular, each subgame $\Psi_i$ is won by the universal player, against our assumption.
	Therefore for each $\tauvec\in \D_{\Xvec}$ there is a subgame $\Psi_i$ that is won by the existential player under the assignment to the remaining outer variables. 
	From Observation~\ref{obs:abs}, subgames $\Psi_i$ can have one of the following structures:
	
	\begin{enumerate}
		\item $\Psi_i= A^\exists_{(-\Yvec)}\xvec_{(-\Yvec)} \leq \bvec^\exists_{\Yvec=\muvec}$ when $\Phi_{l_i}= \exists \Yvec\in \D_{\Yvec}: A^\exists \xvec \leq \bvec^\exists$, or 
		\item $\Psi_i= \Lambda^{(\Yvec=\muvec)}_{\Phi_{l_i}}[\muvec]$ when $\Phi_{l_i}= \exists \Yvec\in \D_{\Yvec} \forall \Zvec\in \D_{\Zvec}: \Lambda$, where $\Lambda$ itself is a QIP.
	\end{enumerate}
	We argue that one of these being feasible means that its associated $\muvec$ is countermove to $\tauvec$ in the original game. 
	If $\tauvec$ allows $\Phi_{l_i}$, then $\Phi_{l_i}$ is in the original games.
	In the first case, $\Psi_i[\tauvec]= A^\exists_{(-\Xvec-\Yvec)}\xvec_{(-\Xvec-\Yvec)} \leq \bvec^\exists_{\Xvec=\tauvec \Yvec=\muvec}$ is feasible so $\muvec$ is a countermove in the original game that satisfies $\Phi_{l_i}[\tauvec]$.
	In the second case  $\Lambda^{(\Xvec=\tauvec, \Yvec=\muvec)}_{\Phi_{l_i}}[\tauvec, \muvec]$ is feasible so $\muvec$ must be a countermove in the original game that satisfies $\Phi_{l_i}[\tauvec]$. Therefore every $\tauvec$ has an existential countermove.
\end{proof}

In the existential case, we can understand soundness through the theoretical existence of proofs. We unfortunately cannot extract proofs at this stage, as we would need to extract cutting planes proofs from \texttt{Gurobi} which is not supported.

\begin{lemma}
	If \rqis returns $\bot$ on multi-game $\exists \Xvec \in \D_{\Xvec}:\{ \Phi_1 \dots \Phi_n\}$  then there is a $\forall$Exp+Cutting Planes refutation of $\exists \Xvec \in \D_{\Xvec}. \Phi_1 \wedge \dots \wedge \Phi_n$. 
\end{lemma}

\begin{proof}
	We can prove this via induction on the quantifier depth of $\exists \Xvec \in \D_{\Xvec}:\{ \Phi_1 \dots \Phi_n\}$.
	
	For the base case, if all $\Phi_1 \dots \Phi_n$ are quantifier free, \wins is invoked. Given the outer quantifier is $\exists$, we solve the integer program~\eqref{prob::ExistEval}, which is the conjunction of all constraint systems $\Phi_1 \dots \Phi_n$. Since cutting planes is a complete refutational system for Integer Programming~\cite{chvatal1984cutting} there is a cutting planes proof of this refutation. 
	
	
	Now consider the inductive step. Suppose some $\Phi_i$  has a quantifier. Then \rqis goes into the CEGAR loop and only returns $\bot$ if the abstraction $\alpha$ returns $\bot$.  We have to show if the abstraction $\exists \Xvec\in \D_{\Xvec} \exists \Zvec^{(\Yvec=\muvec_1)}_{\Phi_{l_1}}\in \D_{\Zvec}  \dots \exists \Zvec^{(Y=\muvec_k)}_{\Phi_{l_k}}\in \D_{\Zvec} : \{\Psi_1 \dots \Psi_k\}$ returns $\bot$ then we can construct a $\forall$Exp+Cutting Planes refutation of the original formula.
	
	By assuming the induction hypothesis we have an $\forall$Exp+Cutting Planes refutation $\pi$ of $\exists \Xvec\in \D_{\Xvec} \exists \Zvec^{(Y=\muvec_1)}_{\Phi_{l_1}}\in \D_{\Zvec}  \dots \exists \Zvec^{(Y=\muvec_k)}_{\Phi_{l_k}}\in \D_{\Zvec} : \{\Psi_1 \dots \Psi_k\}$.
	Note that the quantified variables in each $\Psi_1 \dots \Psi_k$ are different. The only variables they share are in the outer block variables $\Xvec$.

	Consider an individual axiom step $C$ of $\pi$, which involves taking a constraint from $\Psi_i$, and a complete assignment $\betavec$ to the universal variables of $\Psi_i$ that satisfies the universal constraint system. 
	Note that $\Psi_i$ can only be a subgame in the abstraction for one of two reasons:
	
	\begin{enumerate}
		\item $\Psi_i= A^\exists_{(-\Yvec)}\xvec_{(-\Yvec)} \leq \bvec^\exists_{\Yvec=\muvec}$ when $\Phi_{l_i}= \forall \Yvec\in \D_{\Yvec}: A^\exists \xvec \leq \bvec^\exists$, when $\muvec \in \D_{\Yvec}$.
		\item $\Psi_i= \Lambda^{(\Yvec=\muvec)}_{\Phi_{l_i}}[\muvec]$ when $\Phi_{l_i}= \forall \Yvec\in \D_{\Yvec} \exists \Zvec \in \D_{\Zvec}: \Lambda$, when $\muvec \in \D_{\Yvec}$.
	\end{enumerate}
	
	In each case, for every $\Psi_i$ there is a $\Yvec=\muvec$ statement that corresponds to it. 
	We take $\pi$ and rename all the existential variables appearing in the inner blocks to create $\pi'$, we rename $w$ appearing in $\Psi_i$ to be $w^{(\Yvec=\muvec)}_{\Psi_i}$.
	
	In the first case, we can take the single row $j$ from $A^\exists_{(-\Yvec)}\xvec_{(-\Yvec)} \leq \bvec^\exists_{\Yvec=\muvec}$ that was combined with $\betavec$ to get $C$. Note that $\betavec$ must be empty, because there are no universals left.
	We take the same row $j$ of $A^\exists \xvec \leq \bvec^\exists$ and combine it with $\muvec$ (which satisfies the uncertainty set) in an axiom step from our original formula. Then $A_j^\exists \xvec \leq \bvec^\exists_j$ instantiates to $(A^\exists_{(-\Yvec)})_j\xvec_{(-\Yvec)} \leq (\bvec^\exists_{(\Yvec=\muvec)})_j$, exactly the same as $C$ and the transformation in $\pi'$ does not change this, because there are no inner existential variables.
	
	In case 2, the inner part of $\Lambda^{(\Yvec=\muvec)}_{\Phi_{l_i}}$ must have a row $j$ that appears as constraint $$\displaystyle\sum_{k\in [n]:\ Q_k=  \exists} a^\exists_{j,k} x_k + \displaystyle\sum_{k\in [n]:\ Q_k=  \forall} a^\exists_{j,k} x_k \leq (\bvec^\exists_{(\muvec)})_j$$ that combines with $\betavec$ to get axiom $$\displaystyle\sum_{k\in [n]:\ Q_k=  \exists} a^\exists_{j,k} x^{[\betavec]}_k + \displaystyle\sum_{k\in [n]:\ Q_k=  \forall} a^\exists_{j,k} \betavec(x_k) \leq (\bvec^\exists_{(\muvec)})_j.$$ In $\pi' $ we can write it as 
	
	$$\displaystyle\sum_{\substack{k\in [n]\\x_k \in \Xvec}} a^\exists_{j,k} x_k + \displaystyle\sum_{\substack{k\in [n]\\ x_k \in \Zvec^{(\Yvec=\muvec)}_{\Phi_{l_i}}}} a^\exists_{j,k} x_k + \displaystyle\sum_{\substack{k\in [n]:\ Q_k=  \exists\\x_k \text{ inner}}} a^\exists_{j,k} {x_k}_{\Phi_{l_i}}^{(\Yvec=\muvec)[\betavec]} + \displaystyle\sum_{k\in [n]:\ Q_k=  \forall} a^\exists_{j,k} \betavec(x_k) \leq (\bvec^\exists_{(\muvec)})_j .$$
	
	But notice that all $x_k\in \Zvec^{(\Yvec=\muvec)}_{\Phi_{l_i}}$ match the renaming annotation in $\pi'$. Simplifying again gets us
	$$\displaystyle\sum_{\substack{k\in [n]\\ x_k \in \Xvec}} a^\exists_{j,k} x_k+ \displaystyle\sum_{\substack{k\in [n]:\ Q_k=  \exists\\ x \notin \Xvec}} a^\exists_{j,k} {x_k}_{\Phi_{l_i}}^{(\Yvec=\muvec)[\betavec]}  \leq (\bvec^\exists_{(\muvec\sqcup\betavec)})_j. $$
	
	We now take the same row in the inner part of $\Lambda$ and combine it with $\muvec \sqcup \betavec$ to get the axiom 
	$$\displaystyle\sum_{k\in [n]:\ Q_k=  \exists} a_k x^{[\muvec \sqcup\betavec]}_k + \displaystyle\sum_{k\in [n]:\ Q_k=  \forall} a_k \betavec(x_k) \leq (\bvec^\exists)_j .$$
	Note that  $\muvec \sqcup \betavec$ is a disjoint union because $\Yvec$ is already assigned in $\Psi_i$, and each universal block satisfies its domain given by bounds and universal constraint system. Subtracting both sides ends us with the axiom exactly as it was in $\pi'$, because $\Xvec$ variables are not changed under $[\muvec \sqcup \betavec]$ and the $\muvec$ annotation is already present in the $\pi'$  proof. $\pi'$ is therefore a refutation in the original formula.
\end{proof}

\subsection{Optimization}
\subsubsection{Optimization Method 1: Binary Search}
In Section~\ref{sec:prelim}, we introduced the QIP optimization problem, for which a search-based solution approach exists \cite{hartisch2022general}. We now want to utilize the presented expansion-based approach, to obtain another solution tool for the optimization problem. To this end, we assume that a linear objective is given with objective coefficients $\cvec\in\mathbb{Z}^n$. For clarity of presentation, we assume an existential starting player in this case. We have already drawn the connection between the QIP optimization problem and its decision version, where the objective function is moved to the constraints and one asks for the existence of a solution with objective value less than or equal to some value $z$. As all variables are bounded and the objective value only attains integer values, we can compute lower and upper bounds on the objective value. In particular, for the optimal objective value we know $z^\star \in [\min_{\xvec\in\D} \cvec^\top\xvec,\max_{\xvec\in\D} \cvec^\top\xvec]$. Now, we can conduct a binary search to close in on the optimal value as shown in Algorithm~\ref{alg:opt}.

\begin{algorithm}[htb]
\begin{algorithmic}[1]
\Require{QIP optimization problem}
\Ensure{$\bot$, if instance is infeasible and otherwise the optimal objective value.}
\State $LB \gets \min_{\xvec\in\D} \cvec^\top\xvec$
\State $UB \gets \max_{\xvec\in\D} \cvec^\top\xvec$ 
\State $z\gets UB$
\State run \rqis on QIP decision problem with additional constraint $\cvec^\top \xvec \leq z$
\If{$\bot$}
\Return \textit{``Instance is infeasible''}
\EndIf
\While {$UB-LB>0$}
\State $z\gets (LB+UB)/2$
\State run \rqis on QIP decision problem with additional constraint $\cvec^\top \xvec \leq z$
\IfThenElse{feasible}{$UB\gets z$}{$LB\gets z+1$}
\EndWhile
\State \Return $UB$
\end{algorithmic}
\caption{Optimization Framework utilizing \rqis\label{alg:opt}}
\end{algorithm}

\subsubsection{Optimization Method 2: Mixing Methods\label{sec::mixing}}
Repeatedly running the solver in a binary search can end up being more expensive than necessary.
Our working hypothesis is that the existing search-based solver \yasol \cite{hartisch2022general} is good at finding solutions, but slow at verifying optimality, while \rqis suffers from the opposite problem: it is much better at showing inconsistencies from objective value bounds that are too tight, but slow at finding good initial solutions.
Therefore we propose to combine the approaches, by letting \yasol search for good solutions and use \rqis to verify optimality. In particular, every time \yasol finds a new solution---an improved value for $UB$---we can call \rqis and try to verify that no solution with objective value at most $UB-1$ exists by adding the respective constraint on the objective function. Then, if and only if, $UB$ is optimal, no winning move for the existential player will be found. Doing this sequentially, i.e., stopping the search process of \yasol while \rqis tries to verify optimality, of course can be detrimental, as in the early phase of the optimization process, \yasol tends to find better solutions quickly. Thus, calling \rqis for every newly found solution can become inefficient. 
We instead think of initiating a parallel process of \rqis while \yasol continues its search. If \rqis certifies that no better solution exists, \yasol can terminate early. Conversely, if \yasol finds a new solution, the current \rqis process can be terminated and restarted with the updated bound.

It is noteworthy, that if \rqis returns a winning move with an objective function bounded by $z$, it cannot be concluded that a solution with objective value $z$ exists, but only that a solution with some value less than or equal to $z$ exists. Hence, \yasol cannot extract a newly found solution from a call to \rqis that does not verify optimality. 

In this combined approach we also see options to exploit further synergies, where not only \yasol profits from \rqis, but also the other way around. When \yasol continues the search process after finding a new incumbent solution, it would incorporate all information it has gathered during the solution process. Running \rqis with a new bound on the objective function, on the other hand, is essentially starting again from scratch, which we want to avoid. Since constraints are learned during the search process of \yasol, adding them to \rqis is one way to transfer learned information from one solver to the other. Any learned constraint holds information on what variable assignments will not lead to a (improved) solution. It is worth mentioning that \yasol does not explicitly track all its progress through learning constraints, but also implicitly through branching decisions. But each branch of the search tree that is completed without a found (better) solution can be interpreted as a learned constraint. Adding such constraints to the verification instance, has the potential to improve the runtime of \rqis, with the risk of increasing the size of the instance too much, making it harder for the underlying IP solver to solve Problems~\eqref{prob::ExistEval} and \eqref{prob::AllEval}. This issue has similarities to the selection of cutting planes for solving integer programs (see, e.g., \cite{dey2018theoretical}) and further research needs to be done for our special case.

\begin{theorem}
Let $\Pi \phi$ be a QIP with $\Pi$ a quantifier prefix and $\phi$ an IP and assume \yasol is a correct clause learning algorithm for QIP and also complete, in that it will eventually find the optimal solution.  Suppose \yasol learns clauses $C_1 \dots C_n$ on the way to learning a solution with objective function value of $v$. Then $\Pi \phi\wedge F$ is false if and only if $\Pi \phi \wedge C_1 \dots C_n \wedge D \wedge F $ is false, where $F$ is a constraint saying the objective function is strictly less than $v$.

\end{theorem}
\begin{proof}
Suppose \yasol has learned clauses $C_1 \dots C_n$ and a solution with value $v$. Let us consider the QIP feasibility problem $\Pi \phi \wedge C_1 \dots C_n \wedge F$. This QIP is either true, resulting in the existence of a strategy with value less than $v$, or there is no such strategy, rendering the QIP false. If $\Pi \phi \wedge C_1 \dots C_n \wedge F$ is true, then obviously $\Pi \phi \wedge F$ also must be true, as it contains less existential constraints. Now assume $\Pi \phi \wedge C_1 \dots C_n \wedge F$ is false. Assume $\Pi \phi \wedge F$ is true, which means that there exists a solution with objective value strictly less than $v$. As \yasol is complete, it eventually has to find this solution. However, $\Pi \phi \wedge C_1 \dots C_n \wedge F$ being false, implies that \yasol can no longer find the solution if it restarted after learning $C_1 \dots C_n$ (which is an option for \yasol to perform after clause learning). Consequently,  $\Pi \phi \wedge F$ also must be false.
\end{proof}

\section{A New Challenging Problem Class: QRandomParity\label{sec::QRandomParity}}
\subsection{Motivation}
QRandomParity is a combination of QParity, which are known to be hard for QCDCL based QBF solver \cite{BCJ19}, and RandomParity which are hard for CDCL based SAT solvers \cite{chew2024hardness}.

Given an integer $n$ and a random permutation $\sigma$ on $[n]$. Consider the Quantified Boolean problem 
\begin{align*}
\exists x_1 \dots x_n \forall z \exists t_2 \dots t_n \exists s_2 \dots s_n. \ & t_2=(x_1 \oplus x_2)\wedge \dots \wedge t_i=(t_{i-1}\oplus x_i),\dots \wedge\\
&s_2=(x_{\sigma(1)} \oplus x_{\sigma(2)}) \wedge \dots \wedge s_i=(s_{i-1} \oplus x_{\sigma(i)}),\dots \wedge\\
&(z \rightarrow \neg t_n)\wedge(\neg z \rightarrow s_n)
\end{align*}

Both $t_n$ and $s_n$ compute the parity of the $x$ variables but use a different ordering. In particular, $t_n=s_n$ must be fulfilled. If we take the full expansion we get a contradiction, because parity is associative and commutative. These families have been shown hard for CDCL solvers like \cadical in both theory and experiments \cite{chew2024hardness}. This is because of the standard encodings of the parity problems into clauses. Typically one encodes $a= (b \oplus c)$ using four clauses ($\neg a \vee b \vee c)$, $(a \vee \neg b \vee c)$, $(a \vee b \vee \neg c)$, and $(\neg a \vee \neg b \vee \neg c)$.

Note that, extension variables can be used to produce formulas easy for resolution. In this case its has been shown \cite{chew2020sorting} that only $O(n \log n)$ many extension variables are needed before we can get linear size resolution refutation. 
In pseudo-boolean constraints the extension variables come more naturally into a standard encoding of parity, but require a new variable $e$ for each constraint. Therefore $a= (b \oplus c)$ can be represented by the  constraint
$a+b+c=2e$.

This encoding allows a short cutting planes proof. We simply can add all constraints 
to get $x_1 + \dots +x_n + t_n +2t_2 \dots +2t_{n-1} - 2e_2 \dots -2e_n = 0$
 and this can be repeated for $x_1 + \dots +x_n + s_n +2s_2 \dots +2s_{n-1}- 2f_2 \dots -2f_n = 0$.
 
 Subtracting one from the other we get $t_n- s_n + 2 \sum_{i=2}^n \left( t_i -e_i - (s_i - f_i) \right)  = 0$.  The only integer solutions to this equality is to have $t_n=s_n$, as otherwise an odd number would be on the left-hand side of the constraint, and cutting planes finds this via the division rule \cite{gocht2021certifying}. The idea is, that IP solvers, which have access to several preprocessing techniques that utilize constraint aggregation (e.g., see \cite{glover1968surrogate,marchand2001aggregation,wolsey1999integer}), may  be able to handle this as well.
  
 \subsection{ Encodings\label{app::QRP_formulation}}
 
 Given a random permutation $\sigma$ on $[n]$. Then the QRandomParity problem in clausal form (as QBF) is given by the quantification sequence
 $\exists x_1 \dots x_n \forall u \exists t_2 \dots t_n, s_1 \dots s_n$ (with all Boolean variables) and matrix 
 \begin{subequations}
 	\label{prob::QRandomParityQBF}
 	{\small
 		\begin{align}
 			&(\bar{x}_1 \vee x_2 \vee t_2) \ (x_1 \vee \bar{x}_2 \vee t_2)\  (x_1 \vee x_2 \vee \bar{t}_2) \ (\bar{x}_1 \vee \bar{x}_2 \vee \bar{t}_2) \label{eq::QRP_cl1} \\
 			&(\bar{x}_i \vee t_{i-1} \vee t_i) \ (x_i \vee \bar{t}_{i-1} \vee t_{i})\ (x_i \vee t_{i-1} \vee \bar{t}_i) \ (\bar{x}_i \vee \bar{t}_{i-1} \vee \bar{t}_{i})  \text{ for } i=3 \text{ to } n\\
 			&(\bar{x}_{\sigma(1)} \vee x_{\sigma(2)} \vee s_2) \ (x_{\sigma(1)} \vee \bar{x}_{\sigma(2)} \vee s_2)\ (x_{\sigma(1)} \vee x_{\sigma(2)} \vee \bar{s}_2) \ (\bar{x}_{\sigma(1)} \vee \bar{x}_{\sigma(2)} \vee \bar{s}_2)  \\
 			&(\bar{x}_{\sigma(i)} \vee s_{i-1} \vee s_i) \ (x_{\sigma(i)} \vee \bar{s}_{i-1} \vee s_{i})\  (x_{\sigma(i)} \vee s_{i-1} \vee \bar{s}_i) \  (\bar{x}_{\sigma(i)} \vee \bar{s}_{i-1} \vee \bar{s}_{i})  \text{ for } i=3 \text{ to } n \label{eq::QRP_cl4}\\
 			&(\bar{u} \vee \bar{t}_n)\ ({u} \vee {s}_n),
 		\end{align}
 	}
 \end{subequations}
 where each \eqref{eq::QRP_cl1}-\eqref{eq::QRP_cl4} encodes an XOR relation. In linear form (as QIP), all variables are binary and their sequence is given by
 $\exists x_1 \dots x_n \forall u \exists t_2 \dots t_n, d_2 \dots d_n, s_1 \dots s_n, e_2 \dots e_n$, where a linear encoding of the XOR relation is used, needing auxiliary variables $d$ and $e$:
 \begin{subequations}
 	\label{prob::QRandomParityQIP}
 	\begin{align}
 		{x}_1 + x_2 + t_2= 2 d_2 && {x}_{\sigma(1)} + x_{\sigma(2)} + s_2=2 e_2\\
 		t_{i-1} + {x}_i  + t_i= 2 d_i  &&
 		s_{i-1} + {x}_{\sigma(i)}  + s_i= 2 e_i &&\text{ for } i\in \{3,\ldots,n\}\\
 		-u- t_n \geq -1 &&
 		u+ s_n \geq 1,
 	\end{align}
 \end{subequations}

\section{Experimental Evaluation}
For our experiments, we compiled \rqis with \gurobi 12.0 and \yasol, which utilizes a linear programming solver, is compiled with \cplex 22.1. When installing \qfun with the respective SAT solvers, we used the provided script and we used \depqbf 6.01 and the latest version of \zthree.

\subsection{QRandomParity\label{sec::QRP_exp}}
As argued in Section~\ref{sec::QRandomParity} we expect our expansion-based solver to perform well on instances of type QRandomParity,  compared to search-based algorithms like \yasol \cite{hartisch2022general} and \depqbf \cite{lonsing2017depqbf}. We also tested the solver \zthree \cite{de2008z3}, which is capable of handling such instances. Furthermore, we are interested in the comparison of \rqis against state-of-the-art expansion-based solvers from the QBF community such as \qfun \cite{janota2018towards}. \qfun allows the integration of several SAT solvers, and we compiled it once with \cadical \cite{BiereFallerFazekasFleuryFroleyks} and once with \cms \cite{soos2009extending}, which we refer to as \qfunmini and \qfuncms, respectively. Here we expect \qfuncms to perform better on QRandomParity instances, due to the better handling of XOR clauses by \cms compared to other SAT solvers. For \yasol and \rqis, we use the QIP Encoding~\eqref{prob::QRandomParityQIP}, while for the other solvers we use QBF Encoding~\eqref{prob::QRandomParityQBF}. Notably, we tested both encodings with \zthree and observed that it performed better on the QBF formulation than on the QIP formulation. Consequently, we report only its performance on Encoding~\eqref{prob::QRandomParityQBF}. 

For each of the following experiments, we created $100$ instances for each $n$ (only varying in the random permutation of the variables). Experiments were conducted on an AMD Ryzen 9 5900X processor (3.70 GHz) with 128 GB RAM, imposing a $1800$ seconds time limit per instance and restricting each process to a single thread. In a first experiment, for $n\in\{10,12,\ldots,26\}$ we compare all solvers and show the results in a Cactus plot in Figure~\ref{fig::cactus}.
\begin{figure}[htb]
\centering
\includegraphics[width=.65\textwidth]{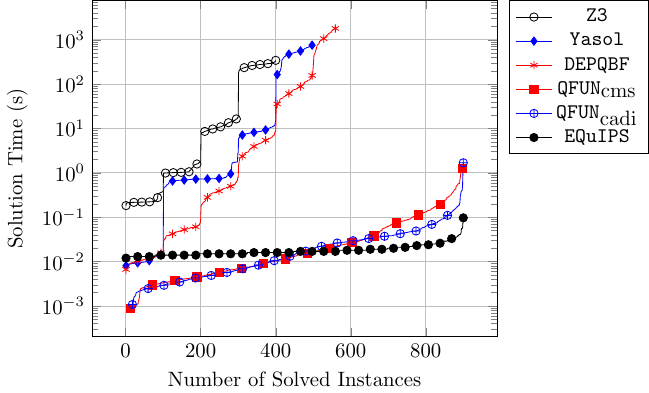}
\caption{Cactus plot. Algorithms that appear further to the right and closer to the bottom solve more instances faster, indicating better performance.\label{fig::cactus}}

\end{figure}
As expected, we can observe the expansion-based solvers outclass the search-based solvers \depqbf and \yasol as well as \zthree. \qfunmini and \qfuncms basically have the same performance, while our approach has the most constant behavior, outperforming all solvers for $n\geq 22$. 

For even larger values of $n$, \qfunmini quickly reaches its limit, incapable of solving these instances before timeout, while \qfuncms still can solve such instances easily. In particular, for $n=100$, no instance was solved by \qfunmini before the timeout. This is likely due to the special handling of XOR formulations by \cms, that \cadical lacks. Hence, for larger values of $n$ we restrict our comparison to \qfuncms and \rqis, as shown in Table~\ref{tab::QRPLargeN}.
\begin{table}[htb]
\centering
\caption{Median runtimes and (if existent) number of not solved QRandomParity instances.  \label{tab::QRPLargeN}}
\begin{tabular}{lllllllllll}
\toprule
$n$&100&200&300&400&500&600&700&800&900&1000\\\midrule
\rqis &0.06& 0.18 & 0.37 & 0.68 & 0.99 & 1.31 & 1.72 & 2.13 & 2.64 & 3.44/10\\
\qfuncms&  1.67 & 2.21 & 2.09 & 2.18 & 2.37 & 2.13/1 &2.51/2 & 2.22/3 & 3.04/4& -/100\\\bottomrule
\end{tabular}
\end{table}

For an increasing value of $n$, our approach consistently outperforms \qfuncms. However, for very large $n$, both solvers sometimes fail to return a solution within the time limit. This is quite surprising, as runtime trends did not indicate a sharp increase in solution times. We further investigated this by analyzing the behavior of \gurobi on the fully expanded problem for large values of $n$. It became evident that for smaller instances, \gurobi could solve them entirely during preprocessing, without initiating a search. However, for larger instances, preprocessing terminated prematurely before infeasibility was detected, forcing the solver into a search phase from which it never returned. For some instances, we were able to fine tune \gurobi parameters, but not to the extend of solving instances with $n\geq 1500$. We suspect a similar phenomenon occurs with \cms. These observations align with our hypothesis that aggregation techniques efficiently prove infeasibility for smaller instances, but as the problem size grows, identifying the right constraints for aggregation becomes increasingly difficult. This, in turn, leads \gurobi to halt preprocessing prematurely and initiate an exhaustive search instead.

For all experiments, \qfun and \depqbf were given instances in the \qdimacs file format. Surprisingly, testing the \qcir format---where we expected \qfuncms to perform even better---resulted in worse performance. Additionally, we evaluated \rqis on the QBF formulation of QRandomParity, where all constraints are clauses. 
While \rqis could still solve instances of size $100$ within seconds, it failed to solve any instance of size 200. This further supports our claim that leveraging the modeling capabilities of linear constraints can be advantageous.

\subsection{Multilevel Critical Node Problem\label{sec::MCN}}
To evaluate our approach on optimization instances, we consider the multilevel critical node problem (MCN) as introduced in \cite{baggio2021multilevel}. Given a directed graph $G=(V,E)$. Two agents act on $G$: The \textit{attacker} selects a set of nodes she wants to infect and the \textit{defender} tries to maximize the number of saved nodes. The defender can \textit{vaccinate} nodes before any infection occurs and \textit{protect} a set of nodes after the attack. An infection triggers a cascade of further infections that propagates via the graph neighborhood, only stopped by vaccinated or protected nodes. For each action (vaccination, infection, protection) a budget ($\Omega$, $\Phi$, $\Lambda$, respectively) exists limiting the number of chosen nodes. For any node $v \in V$ binary variables $z_v$, $y_v$, and $x_v$ are used to indicated its vaccination, infection, and protection, respectively.  Variables $\alpha_v \in \{0,1\}$ indicate whether node $v\in V$ is saved eventually. Only the variables $\pmb{y}$ are universally quantified. Their domain is restricted by a budget constraint, i.e., 
$\U_{\pmb{y}} =\{\pmb{y}\in \{0,1\}^V \mid \sum_{v \in V} y_v \leq \Phi\}$.
Hence, the universal constraint system only contains the single budget constraint. A QIP with polyhedral uncertainty set can be stated as follows:
\begin{subequations}
\label{Model::MCN::QIP}
\begin{align}
\max_{\pmb{z} \in \{0,1\}^V} \min_{\pmb{y} \in \U_{\pmb{y}}} \max_{\substack{\pmb{x} \in \{0,1\}^V\\ \pmb{\alpha} \in \{0,1\}^V}}& \sum_{v \in V} \alpha_v\\
\textnormal{s.t. } \exists \pmb{z} \in \{0,1\}^V \ \forall \pmb{y} \in \U_{\pmb{y}}\  \exists \pmb{x} \in \{0,1\}^V\ \pmb{\alpha} \in \{0,1\}^V: \span \span\\
&\sum_{v \in V} z_v \leq \Omega\\
&\sum_{v \in V} x_v \leq \Lambda\\
&\alpha_v \leq 1+z_v-y_v &&\forall v \in V\label{Model::MCN::QIP::Vaccination}\\
&\alpha_v \leq \alpha_u +x_v +z_v &&\forall(u,v)\in E\label{Model::MCN::QIP::Cascade}
\end{align}
\end{subequations}
Constraint \eqref{Model::MCN::QIP::Vaccination} ensures that infected nodes cannot be saved, unless they were vaccinated and Constraint \eqref{Model::MCN::QIP::Cascade} describes the propagation of the infection to neighboring nodes that are neither vaccinated nor protected. This model corresponds exactly to the trilevel program presented in \cite{baggio2021multilevel} with the key difference, that we are able to directly plug this model into our solver to obtain the optimal solution, without having to dualize, reformulate or develop domain specific algorithms. The same is true for the QIP solver \yasol. 

We compare the performance of \yasol, \rqis (utilizing binary search), and the baseline column- and row-generation approach (\MCNBaggio) from \cite{baggio2021multilevel}. Note that \MCNBaggio is essentially a scenario generation approach that uses dualization techniques to approximate the optimal solution of the adversary problem. In contrast, the approach provided by \yasol and \rqis is much more straightforward and easy to use, requiring only a problem encoding as QIP. Notably, we do not compare our solver to the enhanced techniques from \cite{baggio2021multilevel}, as our goal is to demonstrate the model-and-run potential of QIPs as a baseline for such instances.

The \MCNBaggio algorithm was executed using Python 2.7.18, with mixed-integer linear programs solved via IBM CPLEX 12.9. Experiments were conducted on an AMD EPYC 9474F 48-Core Processor (3.60 GHz) with 256 GB RAM, imposing a two-hour time limit per instance and restricting each process to a single thread. We used the provided instances\footnote{Instances, optimal solutions, and algorithms from \cite{baggio2021multilevel} were provided at \url{https://github.com/mxmmargarida/Critical-Node-Problem}.}, consisting of randomly generated undirected graphs with $|V| \in \{20,40,60,80,100\}$, a density of 5\%, and various budget limit configurations ($\Omega$, $\Phi$, and $\Delta$).
\begin{figure}[htb]
\centering
\includegraphics[width=.65\textwidth]{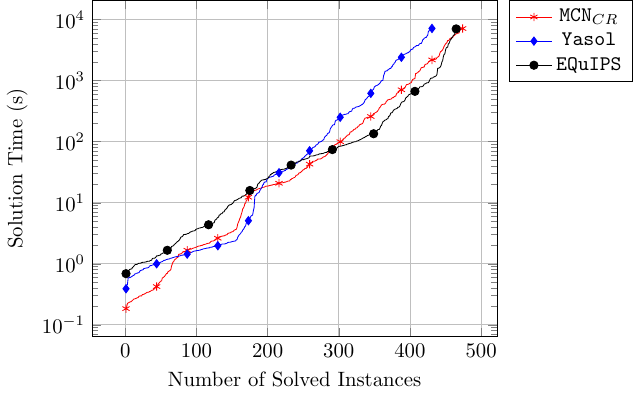}
\caption{Cactus plot for experiments on MCN test set.\label{fig::cactus_mcn}}
\end{figure}
Figure~\ref{fig::cactus_mcn} presents a Cactus plot, revealing key performance trends. Our approach exhibits higher run times on instances that are generally solved quickly. However, \rqis performs well on larger instances with inherently higher run times. Overall, our method solves the second-highest number of instances (465), compared to 431 solved by \yasol and 475 by \MCNBaggio.

\subsection{Further Experiments}
To further evaluate our expansion-based approach, we conducted additional experiments (see Appendix~\ref{app::furtherExperiments} for details). The key findings are as follows. First, using \rqis to verify the optimality of incumbent solutions found during the search process of \yasol is promising. This hybrid approach can be enhanced by incorporating selected learned constraints from \yasol into the \rqis verification instance. Second, \rqis does not always outperform \yasol ---especially on instances with multiple quantifier alternations---a result that aligns with similar observations in QBF. Finally, several heuristic improvements remain to be explored. For example, solving the IP relaxation to obtain an initial winning move shows potential but may increase run times when the full expansion is eventually required.

\section{Conclusion}

In this paper, we presented an expansion-based approach for quantified integer programming, a rarely explored direction in quantified constraint programming. Our method leverages the power of state-of-the-art integer linear programming solvers to handle abstractions---partial expansions of the quantified program.

\rqis offers advantages over existing QIP solvers like \yasol and various QBF approaches. It inherits the benefits of expansion-based solvers over search-based solvers in QBF while also utilizing the modeling flexibility of linear constraints. Additionally, we observe that QBF solvers incorporating XOR reasoning and expansion, such as \qfun with \cms, perform competitively with \rqis.

Our experiments show that \rqis offers advantages in certain cases. We demonstrate performance improvements on different families that other approaches neglect, strengthening the case for combining methods and illustrating how such combinations can be effective in practice.

Future improvements could integrate machine learning techniques to enhance abstraction construction. While \qfun employs decision trees, our integer programming setting allows for numerical machine learning methods such as support vector machines. \added{Additionally, the counterexample generation in \wins can be heuristically guided by adding an objective to the underlying integer programs, e.g., choosing universal assignments that maximize constraint violation or existential assignments that optimize the objective.} To further improve optimization capabilities, strategy extraction as well as directly incorporating optimization aspects into our framework are also promising directions. \added{Furthermore, parallelizing both the exploration of countermoves in each sub‐game and the objective‐value checks during binary search should improve performance.}


\bibliography{compl}

\begin{thebibliography}{10}

\bibitem{baggio2021multilevel}
Andrea Baggio, Margarida Carvalho, Andrea Lodi, and Andrea Tramontani.
\newblock Multilevel approaches for the critical node problem.
\newblock {\em Operations Research}, 69(2):486--508, 2021.

\bibitem{barbosa2019extending}
Haniel Barbosa, Andrew Reynolds, Daniel El~Ouraoui, Cesare Tinelli, and Clark
  Barrett.
\newblock Extending smt solvers to higher-order logic.
\newblock In {\em Automated Deduction--CADE 27: 27th International Conference
  on Automated Deduction, Natal, Brazil, August 27--30, 2019, Proceedings 27},
  pages 35--54. Springer, 2019.

\bibitem{barrett2018satisfiability}
Clark Barrett and Cesare Tinelli.
\newblock Satisfiability modulo theories.
\newblock {\em Handbook of model checking}, pages 305--343, 2018.

\bibitem{benedetti2008quantified}
Marco Benedetti, Arnaud Lallouet, and J{\'e}r{\'e}mie Vautard.
\newblock Quantified constraint optimization.
\newblock In {\em International Conference on Principles and Practice of
  Constraint Programming}, pages 463--477. Springer, 2008.

\bibitem{bertsimas2016multistage}
Dimitris Bertsimas and Iain Dunning.
\newblock Multistage robust mixed-integer optimization with adaptive
  partitions.
\newblock {\em Operations Research}, 64(4):980--998, 2016.

\bibitem{bertsimas2018binary}
Dimitris Bertsimas and Angelos Georghiou.
\newblock Binary decision rules for multistage adaptive mixed-integer
  optimization.
\newblock {\em Mathematical Programming}, 167:395--433, 2018.

\bibitem{bertsimas2003robust}
Dimitris Bertsimas and Melvyn Sim.
\newblock Robust discrete optimization and network flows.
\newblock {\em Mathematical programming}, 98(1):49--71, 2003.

\bibitem{BB21}
Olaf Beyersdorff and Benjamin B{\"{o}}hm.
\newblock Understanding the relative strength of {QBF} {CDCL} solvers and {QBF}
  resolution.
\newblock In James~R. Lee, editor, {\em 12th Innovations in Theoretical
  Computer Science Conference, {ITCS} 2021, January 6-8, 2021, Virtual
  Conference}, volume 185 of {\em LIPIcs}, pages 12:1--12:20. Schloss Dagstuhl
  - Leibniz-Zentrum f{\"{u}}r Informatik, 2021.
\newblock URL: \url{https://doi.org/10.4230/LIPIcs.ITCS.2021.12}, \href
  {https://doi.org/10.4230/LIPICS.ITCS.2021.12}
  {\path{doi:10.4230/LIPICS.ITCS.2021.12}}.

\bibitem{BCCM18}
Olaf Beyersdorff, Leroy Chew, Judith Clymo, and Meena Mahajan.
\newblock Short proofs in qbf expansion.
\newblock In {\em Theory and Applications of Satisfiability Testing--SAT 2019:
  22nd International Conference, SAT 2019, Lisbon, Portugal, July 9--12, 2019,
  Proceedings 22}, pages 19--35. Springer, 2019.

\bibitem{BCJ19}
Olaf Beyersdorff, Leroy Chew, and Mikol{\'{a}}s Janota.
\newblock New resolution-based {QBF} calculi and their proof complexity.
\newblock {\em {ACM} Trans. Comput. Theory}, 11(4):26:1--26:42, 2019.
\newblock \href {https://doi.org/10.1145/3352155} {\path{doi:10.1145/3352155}}.

\bibitem{BiereFallerFazekasFleuryFroleyks}
Armin Biere, Tobias Faller, Katalin Fazekas, Mathias Fleury, Nils Froleyks, and
  Florian Pollitt.
\newblock {CaDiCaL 2.0}.
\newblock In Arie Gurfinkel and Vijay Ganesh, editors, {\em Computer Aided
  Verification - 36th International Conference, {CAV} 2024, Montreal, QC,
  Canada, July 24-27, 2024, Proceedings, Part {I}}, volume 14681 of {\em
  Lecture Notes in Computer Science}, pages 133--152. Springer, 2024.
\newblock \href {https://doi.org/10.1007/978-3-031-65627-9\_7}
  {\path{doi:10.1007/978-3-031-65627-9\_7}}.

\bibitem{bjorner2015conflicts}
Nikolaj Bj{\o}rner, Mikol{\'a}s Janota, and William Klieber.
\newblock On conflicts and strategies in {QBF}.
\newblock In {\em LPAR (short papers)}, pages 28--41, 2015.

\bibitem{bjorner2015playing}
Nikolaj~S Bj{\o}rner and Mikol{\'a}s Janota.
\newblock Playing with quantified satisfaction.
\newblock {\em LPAR (short papers)}, 35:15--27, 2015.

\bibitem{BB23}
Benjamin B{\"{o}}hm and Olaf Beyersdorff.
\newblock {QCDCL} vs {QBF} resolution: Further insights.
\newblock In Meena Mahajan and Friedrich Slivovsky, editors, {\em 26th
  International Conference on Theory and Applications of Satisfiability
  Testing, {SAT} 2023, July 4-8, 2023, Alghero, Italy}, volume 271 of {\em
  LIPIcs}, pages 4:1--4:17. Schloss Dagstuhl - Leibniz-Zentrum f{\"{u}}r
  Informatik, 2023.
\newblock URL: \url{https://doi.org/10.4230/LIPIcs.SAT.2023.4}, \href
  {https://doi.org/10.4230/LIPICS.SAT.2023.4}
  {\path{doi:10.4230/LIPICS.SAT.2023.4}}.

\bibitem{chen2014beyond}
Hubie Chen.
\newblock Beyond {Q}-resolution and prenex form: A proof system for quantified
  constraint satisfaction.
\newblock {\em Logical Methods in Computer Science}, 10, 2014.

\bibitem{chen2004optimization}
Hubie Chen and Martin P{\'a}l.
\newblock Optimization, games, and quantified constraint satisfaction.
\newblock In {\em International Symposium on Mathematical Foundations of
  Computer Science}, pages 239--250. Springer, 2004.

\bibitem{chew2024hardness}
Leroy Chew, Alexis de~Colnet, Friedrich Slivovsky, and Stefan Szeider.
\newblock Hardness of random reordered encodings of parity for resolution and
  cdcl.
\newblock {\em Proceedings of the AAAI Conference on Artificial Intelligence},
  38(8):7978--7986, Mar. 2024.
\newblock URL: \url{https://ojs.aaai.org/index.php/AAAI/article/view/28635},
  \href {https://doi.org/10.1609/aaai.v38i8.28635}
  {\path{doi:10.1609/aaai.v38i8.28635}}.

\bibitem{chew2020sorting}
Leroy Chew and Marijn~JH Heule.
\newblock Sorting parity encodings by reusing variables.
\newblock In {\em International Conference on Theory and Applications of
  Satisfiability Testing}, pages 1--10. Springer, 2020.

\bibitem{chistikov2017complexity}
Dmitry Chistikov and Christoph Haase.
\newblock On the complexity of quantified integer programming.
\newblock In {\em 44th International Colloquium on Automata, Languages, and
  Programming (ICALP 2017)}, volume~80, page~94. Schloss
  Dagstuhl--Leibniz-Zentrum fuer Informatik, 2017.

\bibitem{CM24}
Abhimanyu Choudhury and Meena Mahajan.
\newblock Dependency schemes in cdcl-based {QBF} solving: {A} proof-theoretic
  study.
\newblock {\em J. Autom. Reason.}, 68(3):16, 2024.
\newblock URL: \url{https://doi.org/10.1007/s10817-024-09707-4}, \href
  {https://doi.org/10.1007/S10817-024-09707-4}
  {\path{doi:10.1007/S10817-024-09707-4}}.

\bibitem{chvatal1984cutting}
Va{\v{s}}ek Chv{\'a}tal.
\newblock {\em Cutting-plane proofs and the stability number of a graph}.
\newblock Inst. f{\"u}r {\"O}konometrie und Operations Research, Rhein.
  Friedrich-Wilhelms-Univ., 1984.

\bibitem{de2008z3}
Leonardo De~Moura and Nikolaj Bj{\o}rner.
\newblock Z3: An efficient smt solver.
\newblock In {\em International conference on Tools and Algorithms for the
  Construction and Analysis of Systems}, pages 337--340. Springer, 2008.

\bibitem{dey2018theoretical}
Santanu~S Dey and Marco Molinaro.
\newblock Theoretical challenges towards cutting-plane selection.
\newblock {\em Mathematical Programming}, 170:237--266, 2018.

\bibitem{ederer2017yasol}
Thorsten Ederer, Michael Hartisch, Ulf Lorenz, Thomas Opfer, and Jan Wolf.
\newblock Yasol: an open source solver for quantified mixed integer programs.
\newblock In {\em Advances in Computer Games: 15th International Conferences,
  ACG 2017, Leiden, The Netherlands, July 3--5, 2017, Revised Selected Papers
  15}, pages 224--233. Springer, 2017.

\bibitem{eirinakis2014quantified}
Pavlos Eirinakis, Salvatore Ruggieri, K~Subramani, and Piotr Wojciechowski.
\newblock On quantified linear implications.
\newblock {\em Annals of Mathematics and Artificial Intelligence},
  71(4):301--325, 2014.

\bibitem{ferguson2006relaxations}
Alex Ferguson and Barry O’Sullivan.
\newblock Relaxations and explanations for quantified constraint satisfaction
  problems.
\newblock In {\em International Conference on Principles and Practice of
  Constraint Programming}, pages 690--694. Springer, 2006.

\bibitem{gent2008solving}
Ian~P Gent, Peter Nightingale, Andrew Rowley, and Kostas Stergiou.
\newblock Solving quantified constraint satisfaction problems.
\newblock {\em Artificial Intelligence}, 172(6-7):738--771, 2008.

\bibitem{gerber1995parametric}
Richard Gerber, William Pugh, and Manas Saksena.
\newblock Parametric dispatching of hard real-time tasks.
\newblock {\em IEEE transactions on computers}, 44(3):471--479, 1995.

\bibitem{glover1968surrogate}
Fred Glover.
\newblock Surrogate constraints.
\newblock {\em Operations Research}, 16(4):741--749, 1968.

\bibitem{gocht2021certifying}
Stephan Gocht and Jakob Nordstr{\"o}m.
\newblock Certifying parity reasoning efficiently using pseudo-{B}oolean
  proofs.
\newblock {\em Proceedings of the AAAI Conference on Artificial Intelligence},
  35(5):3768--3777, 2021.

\bibitem{goerigk2021multistage}
Marc Goerigk and Michael Hartisch.
\newblock Multistage robust discrete optimization via quantified integer
  programming.
\newblock {\em Computers \& Operations Research}, 135:105434, 2021.

\bibitem{goerigk2024introduction}
Marc Goerigk and Michael Hartisch.
\newblock An introduction to robust combinatorial optimization.
\newblock {\em International Series in Operations Research and Management
  Science}, 2024.

\bibitem{goerigk2020min}
Marc Goerigk, Jannis Kurtz, and Michael Poss.
\newblock Min--max--min robustness for combinatorial problems with discrete
  budgeted uncertainty.
\newblock {\em Discrete Applied Mathematics}, 285:707--725, 2020.

\bibitem{gurobi}
{Gurobi Optimization, LLC}.
\newblock {Gurobi Optimizer Reference Manual}, 2024.
\newblock URL: \url{https://www.gurobi.com}.

\bibitem{hartisch2020qip}
Michael Hartisch.
\newblock {\em Quantified integer programming with polyhedral and
  decision-dependent uncertainty}.
\newblock PhD thesis, University of Siegen, Germany, 2020.
\newblock URL: \url{https://dspace.ub.uni-siegen.de/handle/ubsi/1705}, \href
  {https://doi.org/http://dx.doi.org/10.25819/ubsi/4841}
  {\path{doi:http://dx.doi.org/10.25819/ubsi/4841}}.

\bibitem{hartisch2021adaptive}
Michael Hartisch.
\newblock Adaptive relaxations for multistage robust optimization.
\newblock In {\em Pacific Rim International Conference on Artificial
  Intelligence}, pages 485--499. Springer, 2021.

\bibitem{hartisch2016quantified}
Michael Hartisch, Thorsten Ederer, Ulf Lorenz, and Jan Wolf.
\newblock Quantified integer programs with polyhedral uncertainty set.
\newblock In {\em Computers and Games: 9th International Conference, CG 2016,
  Leiden, The Netherlands, June 29--July 1, 2016, Revised Selected Papers 9},
  pages 156--166. Springer, 2016.

\bibitem{hartisch2019mastering}
Michael Hartisch and Ulf Lorenz.
\newblock Mastering uncertainty: towards robust multistage optimization with
  decision dependent uncertainty.
\newblock In {\em PRICAI 2019: Trends in Artificial Intelligence: 16th Pacific
  Rim International Conference on Artificial Intelligence, Cuvu, Yanuca Island,
  Fiji, August 26--30, 2019, Proceedings, Part I 16}, pages 446--458. Springer,
  2019.

\bibitem{hartisch2020novel}
Michael Hartisch and Ulf Lorenz.
\newblock A novel application for game tree search-exploiting pruning
  mechanisms for quantified integer programs.
\newblock In {\em Advances in Computer Games: 16th International Conference,
  ACG 2019, Macao, China, August 11--13, 2019, Revised Selected Papers 16},
  pages 66--78. Springer, 2020.

\bibitem{hartisch2022general}
Michael Hartisch and Ulf Lorenz.
\newblock A general model-and-run solver for multistage robust discrete linear
  optimization.
\newblock {\em arXiv preprint arXiv:2210.11132}, 2022.

\bibitem{ignatiev2016quantified}
Alexey Ignatiev, Mikol{\'a}{\v{s}} Janota, and Joao Marques-Silva.
\newblock Quantified maximum satisfiability.
\newblock {\em Constraints}, 21:277--302, 2016.

\bibitem{DBLP:conf/sat/JanotaKMC12}
Mikol{\'a}{\v{s}} Janota, William Klieber, Jo{\~a}o Marques-Silva, and
  Edmund~M. Clarke.
\newblock Solving {QBF} with counterexample guided refinement.
\newblock In Alessandro Cimatti and Roberto Sebastiani, editors, {\em Proc.\
  15th International Conference on Theory and Applications of Satisfiability
  Testing}, volume 7317, pages 114--128. Springer, 2012.

\bibitem{JM15}
Mikol{\'{a}}\v{s} Janota and Joao Marques{-}Silva.
\newblock Expansion-based {QBF} solving versus {Q}-resolution.
\newblock {\em Theor. Comput. Sci.}, 577:25--42, 2015.

\bibitem{janota2018towards}
Mikoláš Janota.
\newblock Towards generalization in {QBF} solving via machine learning.
\newblock {\em Proceedings of the AAAI Conference on Artificial Intelligence},
  32(1), Apr. 2018.
\newblock URL: \url{https://ojs.aaai.org/index.php/AAAI/article/view/12208},
  \href {https://doi.org/10.1609/aaai.v32i1.12208}
  {\path{doi:10.1609/aaai.v32i1.12208}}.

\bibitem{lonsing2017depqbf}
Florian Lonsing and Uwe Egly.
\newblock Depqbf 6.0: A search-based qbf solver beyond traditional qcdcl.
\newblock In {\em Automated Deduction--CADE 26: 26th International Conference
  on Automated Deduction, Gothenburg, Sweden, August 6--11, 2017, Proceedings},
  pages 371--384. Springer, 2017.

\bibitem{LE18alt}
Florian Lonsing and Uwe Egly.
\newblock Evaluating qbf solvers: Quantifier alternations matter.
\newblock In John Hooker, editor, {\em Principles and Practice of Constraint
  Programming}, pages 276--294, Cham, 2018. Springer International Publishing.

\bibitem{lorenz2015solving}
Ulf Lorenz and Jan Wolf.
\newblock Solving multistage quantified linear optimization problems with the
  alpha--beta nested benders decomposition.
\newblock {\em EURO Journal on Computational Optimization}, 3:349--370, 2015.

\bibitem{maggioni2025sampling}
Francesca Maggioni, Fabrizio Dabbene, and Georg~Ch Pflug.
\newblock Sampling methods for multi-stage robust optimization problems.
\newblock {\em Annals of Operations Research}, pages 1--39, 2025.

\bibitem{mamoulis2004algorithms}
Nikos Mamoulis and Kostas Stergiou.
\newblock Algorithms for quantified constraint satisfaction problems.
\newblock In {\em International Conference on Principles and Practice of
  Constraint Programming}, pages 752--756. Springer, 2004.

\bibitem{marchand2001aggregation}
Hugues Marchand and Laurence~A Wolsey.
\newblock Aggregation and mixed integer rounding to solve mips.
\newblock {\em Operations research}, 49(3):363--371, 2001.

\bibitem{matsui2010quantified}
Toshihiro Matsui, Hiroshi Matsuo, Marius~Calin Silaghi, Katsutoshi Hirayama,
  Makoto Yokoo, and Satomi Baba.
\newblock A quantified distributed constraint optimization problem.
\newblock In {\em Proc. 9th Int'l. Conf. on Autonomous Agents and Multiagent
  Systems (AAMAS2010)}, volume~1, pages 1023--1030. Nagoya Institute of
  Technology, 2010.

\bibitem{nguyen2020computational}
Danny Nguyen and Igor Pak.
\newblock The computational complexity of integer programming with
  alternations.
\newblock {\em Mathematics of Operations Research}, 45(1):191--204, 2020.

\bibitem{nightingale2009non}
Peter Nightingale.
\newblock Non-binary quantified {CSP}: algorithms and modelling.
\newblock {\em Constraints}, 14:539--581, 2009.

\bibitem{omer2024combinatorial}
J{\'e}r{\'e}my Omer, Michael Poss, and Maxime Rougier.
\newblock Combinatorial robust optimization with decision-dependent information
  discovery and polyhedral uncertainty.
\newblock {\em Open Journal of Mathematical Optimization}, 5:1--25, 2024.

\bibitem{poss2014robust}
Michael Poss.
\newblock Robust combinatorial optimization with variable cost uncertainty.
\newblock {\em European Journal of Operational Research}, 237(3):836--845,
  2014.

\bibitem{postek2016multistage}
Krzysztof Postek and Dick~den Hertog.
\newblock Multistage adjustable robust mixed-integer optimization via iterative
  splitting of the uncertainty set.
\newblock {\em INFORMS Journal on Computing}, 28(3):553--574, 2016.

\bibitem{reynolds2017solving}
Andrew Reynolds, Tim King, and Viktor Kuncak.
\newblock Solving quantified linear arithmetic by counterexample-guided
  instantiation.
\newblock {\em Formal Methods in System Design}, 51:500--532, 2017.

\bibitem{soos2009extending}
Mate Soos, Karsten Nohl, and Claude Castelluccia.
\newblock Extending sat solvers to cryptographic problems.
\newblock In {\em International Conference on Theory and Applications of
  Satisfiability Testing}, pages 244--257. Springer, 2009.

\bibitem{stergiou2005repair}
Kostas Stergiou.
\newblock Repair-based methods for quantified {CSP}s.
\newblock In {\em International Conference on Principles and Practice of
  Constraint Programming}, pages 652--666. Springer, 2005.

\bibitem{subramani2004analyzing}
K~Subramani.
\newblock Analyzing selected quantified integer programs.
\newblock In {\em International Joint Conference on Automated Reasoning}, pages
  342--356. Springer, 2004.

\bibitem{Thuillier_Siegel_Pauleve_2024}
Kerian Thuillier, Anne Siegel, and Loïc Paulevé.
\newblock {CEGAR}-based approach for solving combinatorial optimization modulo
  quantified linear arithmetics problems.
\newblock {\em Proceedings of the AAAI Conference on Artificial Intelligence},
  38(8):8146--8153, Mar. 2024.
\newblock URL: \url{https://ojs.aaai.org/index.php/AAAI/article/view/28654},
  \href {https://doi.org/10.1609/aaai.v38i8.28654}
  {\path{doi:10.1609/aaai.v38i8.28654}}.

\bibitem{vayanos2012constraint}
Phebe Vayanos, Daniel Kuhn, and Ber{\c{c}} Rustem.
\newblock A constraint sampling approach for multi-stage robust optimization.
\newblock {\em Automatica}, 48(3):459--471, 2012.

\bibitem{verger2008guiding}
Guillaume Verger and Christian Bessiere.
\newblock Guiding search in {QCSP}+ with back-propagation.
\newblock In {\em International Conference on Principles and Practice of
  Constraint Programming}, pages 175--189. Springer, 2008.

\bibitem{wolsey1999integer}
Laurence~A Wolsey and George~L Nemhauser.
\newblock {\em Integer and combinatorial optimization}.
\newblock John Wiley \& Sons, 1999.

\bibitem{DBLP:conf/iccad/ZhangM02}
Lintao Zhang and Sharad Malik.
\newblock Conflict driven learning in a quantified {Boolean} satisfiability
  solver.
\newblock In {\em ICCAD}, pages 442--449, 2002.

\end{thebibliography}

\appendix
\newpage
\section{Appendix}\label{sec:appendix}


\subsection{Further Experiments\label{app::furtherExperiments}}

\subsubsection{Mixing Methods} 
As outlined in Section~\ref{sec::mixing}, there is potential to link search-based and expansion-based approaches. To demonstrate this, we conducted an experiment using multilevel critical node instances. 1. Run \yasol and record the time $t_{\textnormal{opt}}$ at which the optimal solution $z^\star$ is first found (its existence is verified, not its optimality). 2. Run \rqis, separately, with the objective function constraint bound $z^\star +1$ (note the maximization objective function) and record the verification time $t_{\textnormal{ver}}$. 
This hypothetical solver---running \yasol in parallel with \rqis for  verification---with runtime $t_{\textnormal{opt}}+t_{\textnormal{ver}}$ solves 27 more instances than \yasol alone and has strictly lower runtime on 196 instances.

We also tested whether adding learned constraints from \yasol's search benefits the expansion-based solver. We modified \yasol to extract every detected conflict as a constraint until the optimal solution was found (extracting beyond this point might prune the optimal solution). For a single instance, we obtained 73 learned constraints and created three variations of the verification instance: one without added constraints, one with all 73 constraints, and one with three hand-picked constraints. Table~\ref{tab::mix_cons} reports the runtimes and the number of IP calls in the \wins function.
\begin{table}[htp]
\centering
\caption{Comparison of three verification instances.\label{tab::mix_cons}}
\begin{tabular}{llll}
\toprule
instance type& original verification instance & org. + 73 constraints & org. + 3 constraints\\\midrule
runtime & 36.3s & 44.4s & 22.2s\\
calls to IP solver & 1034 & 648 & 550 \\\bottomrule
\end{tabular}
\end{table}
For this instance, incorporating learned constraints reduced the iterations (and thus IP calls) for the expansion-based solver. However, too many constraints may increase IP solver runtime; therefore, selectively transferring the ``most beneficial'' constraints to \rqis can significantly enhance the verification process. Although we expected the 73-constraint instance to have fewer IP calls than the one with three constraints, this discrepancy likely results from testing only a single instance. Averaged over a larger set, more constraints should decrease the number of IP calls.

\subsubsection{Performance of \rqis vs. \yasol on other optimization test sets\label{app::WeVsYasol}}
While the computational experiments in Sections~\ref{sec::QRP_exp} and \ref{sec::MCN} show promising results---suggesting that our expansion-based approach can effectively compete with the search-based solver \yasol ---this cannot be stated as a general conclusion, particularly for optimization instances. We conducted tests on 1800 multistage robust assignment instances from \cite{goerigk2021multistage} and 270 multistage robust scheduling instances from \cite{hartisch2020qip}. \added{The former involve combinatorial matching problems under cost uncertainty, while the latter model aircraft scheduling with uncertain arrival times. Both datasets include instances with up to seven decision stages.} In both cases, instances can be encoded either with or without polyhedral uncertainty, labeled with \QIPPU\ and \QIP, respectively.
\begin{table}[htb]
\centering
\caption{Number of solved instances and median runtimes on different test sets.\label{tab::VsYasol}}
\footnotesize
\begin{tabular}{lccccc|ccccc}
\toprule
&\multicolumn{5}{c|}{solved instances}&
\multicolumn{5}{c}{median run times (seconds)}\\\midrule
&\multicolumn{2}{c}{Assignment}&
\multicolumn{2}{c}{Scheduling}&MCN&\multicolumn{2}{c}{Assignment}&
\multicolumn{2}{c}{Scheduling}&MCN\\
&\QIP & \QIPPU &\QIP & \QIPPU & &\QIP & \QIPPU &\QIP & \QIPPU & \\\midrule
\yasol & 1800 & 1800 & 262 & 264&431& 0.4 & 0.4 & 56.7 &29.0 & 247.8\\ 
\rqis  & 1599 & 1760 & 194 & 247&465 &10.6& 2.7 &	287.9 &	51.3&83.8\\
\bottomrule
\end{tabular}
\end{table}

Table~\ref{tab::VsYasol} presents the number of instances solved within a 1800-second time limit as well as the median runtime, highlighting that \rqis struggles with these problem types. Several factors may explain this behavior. First, for assignment instances without a universal constraint system, the structure requires existentially quantified variables to adapt to any changes in universally quantified variables. As a result, nearly the entire expansion must be constructed before a solution can be determined, drastically increasing computational complexity. Additionally, many of the tested instances feature multiple levels of universally quantified variables, unlike the MCN instances, which contain only a single universal level.

This aligns with expansion-based solvers in the QBF domain, which perform well with few quantifier alternations but struggle as the number of universal levels increase.

\subsubsection{Adapted \wins Function in Case of Empty Abstraction}
Several aspects of Algorithm~\ref{alg:br} allow for a choice between standard techniques and more sophisticated implementations. One such aspect is selecting a winning move for the empty abstraction, occurring at Line~\ref{line::find_tau}. The call to \rqis leads to \wins, which, for $\Q=\exists$, returns any assignment satisfying the domain constraints. In our implementation, we initially returned the lower bounds of existentially quantified variables, avoiding the IP solver. While computationally inexpensive, this can result in an assignment that violates the existential constraints, making the first countermove less meaningful.

An obvious alternative is to use the IP relaxation, where we solve the existential constraint system without considering quantification. The goal is to obtain stronger moves that lead to more relevant countermoves, ultimately reducing the number of subgames to be considered.  Note, that as we observed superior runtimes when using this existential IP relaxation, all reported results so far, are based on this implementation.

However, solving an IP is more computationally expensive than assigning lower bounds. We investigated this trade-off using the same test sets as in Appendix~\ref{app::WeVsYasol}. Results in Table~\ref{tab::IP} show that fewer instances of the assignment test set are solved within the 1800-second time limit with the existential IP relaxation. However, for other test sets, \rqis performance improves, with median runtimes decreasing overall. 

\begin{table}[htb]
\centering
\caption{Number of solved instances and median runtimes, with \rqis using lower bounds (LB) vs. solving the existential IP relaxation (exist. IP) to find a winning move for the empty abstraction.\label{tab::IP}}
\footnotesize
\begin{tabular}{lccccc|ccccc}
\toprule
&\multicolumn{5}{c|}{solved instances}&
\multicolumn{5}{c}{median run times (seconds)}\\\midrule
&\multicolumn{2}{c}{Assignment}&
\multicolumn{2}{c}{Scheduling}&MCN&\multicolumn{2}{c}{Assignment}&
\multicolumn{2}{c}{Scheduling}&MCN\\
&\QIP & \QIPPU &\QIP & \QIPPU & &\QIP & \QIPPU &\QIP & \QIPPU & \\\midrule
LB & 1600 & 1770 & 175 & 223&461 &12.3& 3.8 &	590.9 &	127.8&121.9\\
exist. IP  & 1599 & 1760 & 194 & 247&465 &10.6& 2.7 &	287.9 &	51.3&83.8\\
\bottomrule
\end{tabular}
\end{table}
\end{document}